\documentclass[a4paper,11pt]{article}

\usepackage{amsthm}
\usepackage{amsfonts}
\usepackage{amsmath}
\usepackage{amssymb}

\usepackage{mathtools} 
\usepackage{physics}

\usepackage{dsfont} 
\usepackage{bbold} 
\usepackage{xspace} 
 
\usepackage[usenames,dvipsnames]{xcolor} 
\usepackage{cite}

\usepackage{hyperref}
\hypersetup{
    colorlinks=true,
    linkcolor=black,
    citecolor=black,
    filecolor=black,
    urlcolor=black,
}

\usepackage[normalem]{ulem} 

\allowdisplaybreaks

\usepackage{authblk} 
\usepackage[margin=2cm]{geometry} 
\usepackage{graphicx} 
\usepackage{amsthm}
\usepackage{amsfonts}
\usepackage{amsmath}
\usepackage{amssymb}
\usepackage{amsmath,amsthm,amsfonts}		
\usepackage{graphicx}
\usepackage{mathtools}
\usepackage{authblk}
\usepackage{dsfont}
\usepackage{bbold}
\usepackage[usenames,dvipsnames]{xcolor}
\usepackage{hyperref}						
\hypersetup{
	hidelinks,
    colorlinks = true,
    linkcolor = Red,
    citecolor = Red
}
\usepackage[nameinlink]{cleveref}			
\let\cref\Cref
\usepackage{dsfont}
\newtheorem{theorem}{Theorem}
\newtheorem{lemma}[theorem]{Lemma}
\newtheorem*{lemma*}{Lemma}
\newtheorem{claim}[theorem]{Claim}

\newtheorem{corollary}[theorem]{Corollary}

\theoremstyle{definition}
\newtheorem{definition}[theorem]{Definition}

\newcommand{\CC}{\mathcal{C}}

\newcommand{\WW}{\mathcal{W}}
\newcommand{\XX}{\mathcal{X}}

\newcommand{\binomial}{\mathrm{Binomial}}

\DeclareMathOperator{\bbE}{\mathbb{E}}

\newcommand{\st}{\ \middle| \ } 
\newcommand{\cpa}[1]{\left\{ #1 \right\}} 
\newcommand{\pa}[1]{\left( #1 \right)} 
\newcommand{\hmaj}{\(h\)-majority\xspace} 
\newcommand{\threemaj}{\(3\)-majority\xspace} 
\newcommand{\twochoices}{\(2\)-choices\xspace} 
\newcommand{\undecided}{undecided-state\xspace}
\newcommand{\conf}{\mathbf{c}} 
\newcommand{\bias}{\mathrm{B}} 

\newcommand{\pr}[1]{\Pr\left(#1\right)}
\newcommand{\Var}{\mathrm{Var}} 

\DeclareMathOperator*{\expect}{\mathbb{E}}
\newcommand{\winties}[1]{\WW_{#1,\mathrm{ties}}}
\newcommand{\winstrict}[1]{\WW_{#1,\mathrm{strict}}}

\definecolor{myOrange}{rgb}{0.8,0.4,0}

\title{On the $h$-majority dynamics with many opinions}
\author[1]{Francesco d'Amore}
\author[2]{Niccolò D'Archivio}
\author[3]{George Giakkoupis}
\author[2]{Emanuele Natale}  

\affil[1]{Gran Sasso Science Institute, Italy}
\affil[2]{INRIA, COATI, Université Côte d'Azur, France}
\affil[3]{INRIA Rennes}

\date{}

\begin{document}

\maketitle

\begin{abstract}
    We present the first upper bound on the convergence time to consensus of the well-known $h$-majority dynamics with $k$ opinions, in the synchronous setting, for $h$ and $k$ that are both non-constant values.
    We suppose that, at the beginning of the process, there is some initial \textit{additive bias} towards some plurality opinion, that is, there is an opinion that is supported by $x$ nodes while any other opinion is supported by strictly fewer nodes.
    We prove that, with high probability, if the bias is $\omega(\sqrt{x})$ and the initial plurality opinion is supported by at least $x = \omega(\log n)$ nodes, then the process converges to plurality consensus in $O(\log n)$ rounds whenever $h = \omega(n \log n / x)$.
    A main corollary is the following: if $k = o(n / \log n)$ and the process starts from an almost-balanced configuration with an initial bias of magnitude $\omega(\sqrt{n/k})$ towards the initial plurality opinion, then any function $h = \omega(k \log n)$ suffices to guarantee convergence to consensus in $O(\log n)$ rounds, with high probability.
    Our upper bound shows that the lower bound of $\Omega(k / h^2)$ rounds to reach consensus given by Becchetti et al.\ (2017) cannot be pushed further than $\widetilde{\Omega}(k / h)$.
    Moreover, the bias we require is asymptotically smaller than the $\Omega(\sqrt{n\log n})$ bias that guarantees plurality consensus in the $3$-majority dynamics: in our case, the required bias is at most any (arbitrarily small) function in $\omega(\sqrt{x})$ for any value of $k \ge 2$.
\end{abstract}

\section{Introduction}

In this paper, we study the convergence time of the well-known \hmaj dynamics.
In brief, the setting is the following:
We are given a network of \(n\) agents that interact with each other in synchronous rounds.
At round 0, each agent supports one out of \(k\) possible opinions in the set \([k]\):  
At round \(i > 0\), each agent \(v\) samples uniformly at random \(h\) neighbors \(u_1, \ldots, u_h\) (with repetition) which, in turn, send their opinions \(x_1, \ldots, x_h\) to \(v\).
Then, \(v\) adopts as its own opinion the mode of \(\{x_1, \ldots, x_h\}\), breaking ties uniformly at random.
Let \(\conf_t(i)\) count the number of agents supporting opinion \(i\) after round number \(t\).
The \hmaj dynamics gives rise to a random process on the system graph, which we call the \hmaj process.
The question is to determine the number of communication rounds that the process needs to reach \textit{consensus}, i.e., a configuration where all agents agree on some opinion \(i \in [k]\).
Furthermore, we are interested in \textit{plurality consensus}:
Suppose that at the beginning of the process there is a \textit{plurality opinion}, that is, an opinion \(i^\star\) such that \(\conf_0(i^\star) > \conf_0(j)\) for all \(j \neq i^\star\).
Plurality consensus is consensus on the opinion \(i^\star\), and one wishes to quantify the size of the initial bias that guarantees to obtain plurality consensus with high probability.\footnote{The expression \textit{with high probability} (in short, w.h.p.) refers to an event that holds with probability at least \(1 - n^{-c}\), where \(c > 0\) is any constant that is independent of the size \(n\) of the distributed system.}

\subsection{Our contribution}

We give the first non-trivial upper bound to the consensus time of the \hmaj process with \(k\) opinions for non-constant \(h\) and \(k\).
Let us denote the difference \(\max_{i} \min_{j \neq i} \{\conf_t(i) - \conf_t(j)\} \) by \(\bias_t\) and refer to it as the \textit{bias} of the configuration \(\conf_t\) at time \(t\).
In a nutshell, our main result is the following.

\begin{theorem}\label{thm:intro:main}
    Consider the \hmaj dynamics with \(k\) opinions in the complete graph of \(n\) nodes with self-loops \footnote{Assuming self-loops is standard in the literature, as it simplifies the analysis without affecting the possible outcomes.}.
    Let \(\lambda_1, \lambda_2, \lambda_3\) be large enough positive constants.
    Assume that opinion 1 is the initial plurality opinion with \(\bias_0 \ge \lambda_1 \sqrt{\conf_0(1)}\) and \(\conf_0(1) \ge \lambda_2 \log n\).
    If \(h \ge \lambda_3 n \log n / \conf_0(1)\), then, with high probability, the \hmaj process reaches plurality consensus within \(O(\log n)\) rounds.
\end{theorem}

Let us discuss \cref{thm:intro:main} a bit more in details. 
First, notice that it always holds \(\conf_0(1) \ge n/k\) since \(k \ge n / \conf_0(1)\).
In the following, we say that a configuration is almost-balanced if there exists a subset \(I \subseteq [k]\), with \(\abs{I} = O(1)\), such that \(\conf_0(i) = \Theta(\conf_0(1))\) for all \(i \in [k]\setminus I \), assuming that opinion \(1\) is the plurality opinion.
Notice that, in this case, it must hold that \(\conf_0(1) = \Theta(n/k)\) and \(\conf_0(i) = \Theta(n/k)\) for all \(i \in [k]\setminus I\).

\subsubsection*{\texorpdfstring{Case 1: \(k = O(n/ \log n)\).}{Case 1}}
Here we have two sub-cases depending on how we fix \(\conf_0(1)\).
In the first sub-case, we fix \(\conf_0(1) = \Theta(n/k)\) such that \(\conf_0(1) \ge \lambda_2 \log n\) (which is always possible).
Here, we must select some large enough \(h = \Theta(n \log n / c_0(1)) = \Theta(k \log n)\) in order to apply \cref{thm:intro:main}, with the required bias that is minimum, that is, \(\bias_0 = \Theta(\sqrt{n / k})\): observe that the configuration is necessarily almost-balanced.
If, instead, \(\conf_0(1) = \omega(n/k)\), it is guaranteed that \(\conf_0(1) = \omega(\log n)\).
In this case, the configuration \textit{cannot} be almost-balanced because the required bias is \(\bias_0 = \Theta(\sqrt{\conf_0(1)}) = \omega(\sqrt{n/k})\).
Here, we can pick \(h = \Theta(n \log n / \conf_0(1)) = o( k \log n)\).

\subsubsection*{\texorpdfstring{Case 2: \(k = \omega(n/ \log n)\).}{Case 2}}
Since \(k = \omega(n / \log n)\) and \(\conf_0(1) = \Omega(\log n)\) by hypothesis, we \textit{cannot} have almost balanced configurations: we can always choose some \(h = o(k \log n)\) and we will have that the minimum bias is \(\bias_0 = \omega(\sqrt{n/k})\). Note that in this setting, $h$ may be larger than $n$, which we do not exclude as the samples can contain repetitions.

\subsubsection*{Observations.}

Despite the assumption \(h \ge \lambda_3 n \log n / \conf_0(1)\) is undoubtedly a strong one (that is, \(h\) needs to be at least roughly \(k \log n\) for almost-balanced configurations), in \cref{sec:intro:sketch} we show how the analysis for such regime of \(h\) is already quite involved and we discuss what the main difficulties to reduce such assumption are.
Also, in \cref{sec:intro:sketch} we will highlight why we believe that \(h \sim n/\conf_0(1)\) is a ``threshold'' value, since the two settings \(h \ll n/\conf_0(1) \) and \(h \gg n/\conf_0(1) \) seem to need complementary approaches, which is also why our analysis cannot be adapted to the case \(h \ll n/\conf_0(1) \).
However, since \(h\) is a ``static'' parameter of the \hmaj dynamics, in order to approach the case \(h < \lambda_3 n \log n /\conf_0(1) \), one has to show only that there is a time \(t >0\) such that \(\conf_t(1) \ge \lambda_3 n \log n / h\) and that opinion \(1\) is still the plurality opinion with \(B_t\) that meets the hypothesis of \cref{thm:intro:main}, which implies that our analysis applies.

Prior to this work, the only result for the \hmaj dynamics with non-constant \(h\) and \(k\) was a lower bound of \(\Omega(k/h^2)\) rounds to reach consensus given by \cite{becchettiSimpleDynamicsPlurality2017} that holds with high probability whenever \(\max_{i \in [k]}\{\conf_0(i)\} \le 3n/(2k)\).
Hence, \cref{thm:intro:main} shows that the lower bound cannot be pushed to \(\widetilde{\Omega}(k/h)\) (where the notation \(\widetilde{\Omega}\) hides \(\text{polylog}(n)\) factors), at least as long as \(\max_{i \in [k]}\{\conf_0(i)\} \le 3n/(2k)\) is satisfied, e.g., when \((n - \bias_0)/(\conf_0(1) - \bias_0) \le k \le 3n/(2\conf_0(1))\), which is always realized when \(\max\{(n + (k-1)\bias_0)/k, \lambda_2 \log n\} \le \conf_0(1) \le 3n/(2k)\) and \(k \le 3n/(2\lambda_2 \log n)\) for any value of \(\bias_0\) meeting the hypothesis of \cref{thm:intro:main}.

Furthermore, as we will discuss in the related works, in the context of opinion dynamics, the minimum required bias in any other opinion dynamics is always asymptotically larger than our minimum required bias.
Interestingly, the \hmaj dynamics can amplify a very small bias even when \(h = O(k \log n)\) (this is easy to prove for very large \(h\) as we argue in \cref{sec:intro:sketch}, but it is not at all trivial for our parameters).

A final remark that we stress is about the message complexity. 
According to \cref{thm:intro:main},
the message complexity required to reach consensus is \(O(h n \log n)\), which can be as high as \(O(k n \log^2 n ) \) in almost-balanced configurations. 
By recent work on the \threemaj dynamics (see \cref{sec:related-works}), the message complexity of the \threemaj dynamics is \(O(k n \log n)\).
We believe that the gap between these two quantities is due to the fact that the analysis of the \hmaj dynamics is not yet tight.

\subsection{Previous results}\label{sec:related-works}

There is a vast body of research that investigated the \hmaj dynamics when either \(h\) or \(k\) are constant values, both in the synchronous and in the asynchronous settings \cite{becchettiSimpleDynamicsPlurality2017,BerenbrinkCEKMN17,BerenbrinkCGHKR22,GhaffariL18,CooperMRSS25,Shimizu2025}.
Before exploring the results in the two settings, let us discuss the work \cite{BerenbrinkCGHKR22}.
The authors compared the power of the \hmaj dynamics and that of the \((h+1)\)-majority dynamics with \(k = 2\) opinions.
They proved that, starting from the same configuration, the consensus time \(T_{h+1}\) of the \((h+1)\)-majority dynamics is stochastically dominated by the consensus time \(T_h\) of the \hmaj dynamics in the following sense: for each \(t > 0\), \(\Pr(T_{h+1} \le t) \ge \Pr(T_h \le t)\) in both the synchronous and the asynchronous settings.
However, to date no such extension to the case of \(k > 2\) opinions is known.

\subsubsection*{Synchronous setting.}
Let us specify that all the results we mention in this subsection and the next two ones hold in the complete graph of \(n\) nodes with self-loops.
The work \cite{BecchettiCNPT16} proved an upper bound of \(O\pa{(k^2 \sqrt{\log n} + k \log n)(k + \log n)}\) rounds to reach consensus that holds w.h.p., provided that \(k \le n^{\alpha}\) for a suitable positive constant \(\alpha < 1\).

In \cite{becchettiSimpleDynamicsPlurality2017}, the authors showed that the synchronous \threemaj dynamics  with \(k\) opinions converges in \(O(\min\{k, (n/\log n)^{\frac 13}\} \log n)\) with high probability, provided that the bias of the initial configuration is at least \(c \sqrt{\min\{2k, (n/\log n)^{\frac 13}\} n \log n}\) for some constant \(c > 0\).
Moreover, the authors provided a lower bound of \(\Omega(k\log n)\) on the convergence time to consensus, w.h.p., whenever the initial configuration is sufficiently balanced, that is, \(\max_{i \in [k]}\{\conf_0(i)\} \le n/k + (n/k)^{1 - \varepsilon}\) for some \(\varepsilon > 0\) and \(k \le (n/\log n)^{1/4}\).

The work \cite{BerenbrinkCEKMN17} compared the synchronous \threemaj dynamics with the synchronous \twochoices dynamics.
The \twochoices dynamics works as follows: 
At each round, each agent picks two neighbors (with repetition) u.a.r. 
If the two sampled neighbors support the same opinion, the agent adopts that opinion.
Otherwise, the agent keeps its own opinion.
The authors first proved a generic lower bound of \(\Omega(\min\{k, n/\log n\})\) rounds to reach consensus starting from the initial perfectly balanced configuration, w.h.p.
Furthermore, they proved that the \threemaj  dynamics works better in symmetric configurations (i.e., with no initial bias) when, e.g., \(\max_{i \in [k]}\{\conf_0(i)\} = O(\log n)\).
In particular, w.h.p., the \threemaj takes time at most \(O(n^{3/4} \log^{7/8} n)\) to reach consensus w.h.p., regardless of any other hypothesis on the initial configuration, while the \twochoices needs time \(\Omega(n / \log n)\) whenever \(\max_{i \in [k]}\{\conf_0(i)\} = O(\log n)\). 
This was the first work to notice that, for a large number of opinions, the \threemaj dynamics is polynomially (in \(k\)) faster than the \twochoices dynamics.

The work \cite{GhaffariL18} improved upon \cite{becchettiSimpleDynamicsPlurality2017} and showed that, for the \twochoices with \(k = O(\sqrt{n / \log n})\) and for the \threemaj with \(k = O(n^{1/3}/\sqrt{\log n})\), the convergence time to consensus is \(O(k \log n)\), with high probability.
Notice that this upper bound is tight according to the lower bound by \cite{becchettiSimpleDynamicsPlurality2017}, at least as long as \(k \le (n/\log n)^{1/4}\).
Furthermore, the authors showed that the convergence time on the \threemaj dynamics is \(O(n^{2/3} \log^{3/2} n)\) with high probability, regardless of the number of opinions.

Very recently, \cite{Shimizu2025} settled almost tightly the complexity of both the \threemaj and the \twochoices dynamics.
The authors proved that, w.h.p., the \threemaj dynamics reaches consensus in \(O(k\log n)\) rounds whenever \(k = o(\sqrt{n}/\log n)\), while it takes time \(O(\sqrt{n} \log^2 n )\) for other values of \(k\).
Furthermore, \cite{Shimizu2025} proved that plurality consensus is ensured w.h.p.\ as long as the initial bias is \(\bias_0 = \omega(\sqrt{n \log n})\).
As for the \twochoices, they showed that, w.h.p., the dynamics reaches consensus in \(O(k\log n)\) rounds whenever \(k = o(n/\log^2 n)\), while it takes time \(O(n \log^3 n)\) otherwise.
In this case, plurality consensus is ensured w.h.p.\ as long as the initial bias is \(\bias_0 = \omega(\sqrt{\conf_0(1) \log n})\).\footnote{This is the only example of initial bias that gets ``close enough'' to what we require in \cref{thm:intro:main}, but with an exceeding \(\sqrt{\log n}\) factor.}
These results almost match the generic lower bound given by \cite{BerenbrinkCEKMN17}, up to logarithmic factors.

\subsubsection*{Asynchronous setting.}
The only works that analyzed the \threemaj dynamics in the asynchronous setting are \cite{CooperMRSS25,BerenbrinkCGHKR22}. 
In \cite{BerenbrinkCGHKR22}, the authors consider the binary opinion case and show that the convergence time of the asynchronous \threemaj dynamics is \(O(n \log n)\) rounds, w.h.p., and that a bias of \(\Theta(\sqrt{n \log n})\) is sufficient to ensure plurality consensus, w.h.p.
The authors of \cite{CooperMRSS25} showed that the convergence time is \(O(\min\{kn\log^2 n, n^{3/2} \log^{3/2} n\})\), w.h.p., no matter the number of initial opinions. 
They also provided a generic lower bound of \(\Omega(\min\{kn, n^{3/2}/\sqrt{\log n}\})\) rounds to reach consensus (starting from balanced configurations), w.h.p.
The work \cite{CooperMRSS25} (which came before \cite{Shimizu2025}) was the first to establish exactly how the linear-in-\(k\) dependence in the consensus time of the \threemaj dynamics breaks when the number of opinions grows over \(\sqrt{n}\), in which case the consensus time is sublinear in the number of opinions. 
The reader may observe that the asynchronous setting has exactly the same convergence time as the synchronous model, but for a multiplicative factor \(n\).
This is to be expected: 
In the asynchronous setting, in a round only one agent (sampled u.a.r.) updates its state, and hence we need roughly \(n\) rounds to activate all agents at least once (up to polylogarithmic factors).
For processes with small variance (smaller w.r.t.\ the voter model), convergence times from the synchronous to the asynchronous setting usually scale with such a factor.
Note that this is not true for processes with large variance \cite{BecchettiCPTVZ24}.

\subsubsection*{On the \hmaj dynamics.}
As previously discussed, \cite{BerenbrinkCGHKR22} established the ``hierarchy'' of the \(h\)-majority dynamics when the number of opinions is \(k=2\), by showing that the \((h+1)\)-majority is ``faster'' than the \hmaj (both in the synchronous and asynchronou settings),
and  \cite{becchettiSimpleDynamicsPlurality2017} exhibited the lower bound of \(\Omega(k/h^2)\) rounds to reach consensus (that holds w.h.p.).
Apart for the two aforementioned works, there is no theoretical result on the \hmaj dynamics for \(h \gg 1\), which makes it one of the open problems in consensus dynamics.
\cite{Shimizu2025} put the \hmaj dynamics in the ``Open Question'' section of their recent work, and wondered whether the techniques developed in \cite{Shimizu2025} applied to the \hmaj dynamics as well.
We do not know the answer to this question: 
As we will discuss in \cref{sec:intro:sketch}, in the \hmaj the main difficulty is to compute the expectation of \(\conf_{t+1}(i)\) conditional on the whole configuration \(\conf_t\) at the previous rounds, while this is straightforward when \(h = 3\) and is at the core of every paper analyzing the \threemaj dynamics \cite{BecchettiCNPT16,becchettiSimpleDynamicsPlurality2017,GhaffariL18,Shimizu2025,CooperMRSS25}.
More in general, there is no non-trivial upper bound to the consensus time of the \hmaj dynamics with \(k \gg 1\) opinions.
In this paper, we take a first step into the analysis of the \hmaj dynamics.
As we will show, our main effort basically is providing a lower bound to \(\expect[\conf_{t+1}(1) - \conf_{t+1}(2) \mid \conf_t]\), supposing that opinion 1 is the plurality opinion at round \(t\), and opinion \(2\) is the opinion supported by the second-largest community.

\subsubsection*{Other related works.}
The \threemaj, the \hmaj, and other majority-based dynamics have been investigated also in other settings.
For example, in presence of communication noise (which corrupts the exchanged messages) or in presence of stubborn agents \cite{DAmoreZ22,dAmoreZ25}, or when some opinion is preferred among the others, that is, there is some probability that an agent, instead of running the protocol, spontaneously switches to the preferred opinion with some fixed probability \cite{CrucianiMQR23,LesfariGP22}.

Other opinion dynamics for the (plurality) consensus problem have been investigated both in the synchronous and asynchronous settings.
One of the most prominent is the \undecided dynamics.
In the \undecided dynamics there is an extra opinion, the 
\textit{undecided} opinion.
The update rule works as follows:
A node samples one neighbor u.a.r.\ and pulls its opinion.
If the received opinion is different from the supported one, the node becomes undecided.
If the node is undecided, it just copies whatever opinion receives.
The \undecided dynamics has been studied both in the synchronous setting and in the population protocol model (asynchronous setting) \cite{AngluinAE08,ClementiGGNPS18,BecchettiCNPS15,DAmoreCN20,DAmoreCN22,AmirABBHKL23,berenbrink2024,BankhamerBBEHKK22} and performs roughly the same as the \threemaj dynamics for a small number of opinions (while its analysis for the unconstrained general case is still open).
More in general, work on opinion dynamics fits into the recent trend in the distributed computing community of drawing inspiration from social and biological systems.
Other then the consensus problem, researchers have analyzed distributed search problems but also more generic problems (MIS, leader election, etc.) in extremely weak computational models (such as the beeping model or the stone-age model) \cite{KormanV23,FeinermanK17,ClementiDGN21,GiakkoupisZ23,GiakkoupisTZ24,DArchivioV24,FuggerNR24,vacus-ziccardi2025,FraigniaudKR16,FeinermanHK14}.

A special mention goes to \cite{FraigniaudN19}.
Here, the authors considered noisy distributed systems, i.e., systems in which exchanged messages can be corrupted and changed with some positive probability, a setting that takes inspiration by biological societies, where environmental factors can disturb communication.
They characterized the type of noise for which the tasks of \textit{information spreading} and plurality consensus can be solved, when the number of opinions is constant.\footnote{Information spreading is defined as follows: 
In the distributed system there is only one node that is informed about one out \(k\) opinions, while all other nodes are not informed. 
The task is to design a protocol that informs all nodes as fast as possible.}
The protocol designed by \cite{FraigniaudN19} to solve plurality consensus relies on some majority-based rule, and proves a technical result that is at the heart of our analysis. 
We will discuss this in more details in \cref{sec:intro:sketch}.
We nevertheless stress the fact that \cite{FraigniaudN19} only considers the case where the number of opinions \(k\) is a constant.

\section{Preliminaries}

We work on a complete graph of \(n\) nodes with self-loops.
Initially, each node supports one out of \(k\) opinions in the set \([k]\).
Time is synchronous and controlled by some global clock.
At each round, in parallel, every node \(v\) samples \(h\) neighbors u.a.r.\ with repetition and pulls their opinions. 
Then, \(v\) adopts the unique majority opinion, if any, breaking ties u.a.r.\ among the opinions that are in the majority.
For each opinion \(i \in [k]\), we denote by \(\conf_t(i)\) the number of nodes supporting opinion \(i\) at time after performing the update-rule at time \(t\).
Let \(\conf_t = (\conf_t(1), \ldots, \conf_t(k))\) denote the \textit{configuration} of the system at time \(t\).
Notice that the random variable \(\conf_t\) defines a Markov chain: for any sequence of configurations \(y_1, \ldots, y_t \in [n]^k\), it holds that \(\Pr(\conf_t = y_t \mid \cap_{i = 1}^{t - 1} \{\conf_i = y_i\}) = \Pr(\conf_t = y_t \mid \conf_{t-1} = y_{t-1})\) .

Without loss of generality, we assume that at the beginning of the process we have \(\conf_0(1) \ge \conf_0(2) \ge \ldots \conf_0(k)\).
At any time \(t\), we call \textit{plurality opinion} at time \(t\) the opinion \(i\) such that \(\conf_t(i) \ge \conf_t(j)\) for all \(j \in [k]\). 
At any time \(t\), the \textit{bias} of a configuration \(\conf_t\) is the quantity \(\bias_t = \max_{i \in [k]} \min_{j \neq i} \{\conf_t(i) - \conf_t(j)\}\), that is, of how many nodes the community supporting the plurality opinion exceeds the community supporting the second largest one.
Note that \(\bias_t = 0\) if there are multiple plurality opinions, while \(\bias_t > 0\) the plurality opinion is unique.

We work from the perspective of a single node.
When a node \(v\) samples \(h\) neighbors u.a.r., we are drawing from a multinomial distribution.
In particular, given a configuration \(\conf_t\) at time \(t\) and a node \(v\), let \(\mathbf{X}^{(t)}(v) = (X_1^{(t)}(v), \ldots, X_k^{(t)}(v) )\) be a multinomial random variable where \(X_i^{(t)}(v) \sim \text{Bin}(n,p_i^{(t)})\), with \(p_i^{(t)} = \conf_t(i)/n\), counts the number of neighbors supporting opinion \(i\) that \(v\) samples at time \(t\): it must hold that \(\sum_{i = 1}^k X_i^{(t)}(v) = n\), so the variables are negatively correlated.
Let us denote the event that \(v\) supports opinion \(i\) at the end of round \(t\) by \(\WW_i^{(t)}(v)\).

For the rest of the paper, we will use only \(p_1^{(t)}, \ldots, p_k^{(t)}\), since these are the probabilities defining the multinomial distribution of interest.
We define the \textit{normalized bias} to be \(\delta_t = \max_{i \in [k]} \min_{j \neq i} \{p_i^{(t)} - p_j^{(t)}\} \), which is equal to \(\bias_t /n\).
When the round we are referring to is clear from the context, we only write \(p_1, \ldots, p_k\), \(\delta\), \(\mathbf{X}(v) = (X_1(v), \ldots, X_k(v) )\), and \(\WW_i(v)\), omitting the dependency on \(t\), to refer to \(p_1^{(t)}, \ldots, p_k^{(t)}\), \(\delta_t\), \(\mathbf{X}^{(t)}(v) = (X_1^{(t)}(v), \ldots, X_k^{(t)}(v) )\), and \(\WW_i^{(t)}(v)\).
Also, since in any given round \(t\) the random variables \(\{\mathbf{X}(v)\}_{v \in V}\) are i.i.d., and the events \(\{\WW_i(v)\}_{v \in V}\) are mutually independent, when analyzing the distribution of some \(\mathbf{X}(v)\) or the probabilities of some \(\WW_i(v)\), we omit the dependency on \(v\) and only write \(\mathbf{X} = (X_1, \ldots, X_k)\) and \(\WW_i\).

From now on, we will refer to the normalized bias \(\delta_t\) simply as the \textit{bias} of the configuration, as we will not make use of \(\bias_t\) anymore.
With this notation, our main theorem can be rewritten as follows:
\begin{theorem}\label{thm:prelim:main}
    Let \(\lambda_1, \lambda_2, \lambda_3\) be large enough positive constants.
    Assume that opinion 1 is the initial plurality opinion with \(\delta_0 \ge \lambda_1 \sqrt{p_1/n}\) and \(p_1 \ge \lambda_2 \log n / n\).
    If \(h p_1 \ge \lambda_3 \log n\), then, with high probability, the \hmaj process reaches plurality consensus within \(O(\log n)\) rounds.
\end{theorem}

\section{Overview of our analysis and technical challenges}\label{sec:intro:sketch}

The whole analysis is an effort to estimate the quantity \(\Pr(\WW_1) - \Pr(\WW_2)\) from below (assuming that 
\(p_1 \ge p_2 \ge \ldots \ge p_k\)), and use it to get a lower bound on the expect round-by-round growth of the bias \(\delta\).
Notice that, for any given configuration \(\conf_t\) at time \(t\), it holds that
\begin{align}
    \expect[\delta_{t+1} \mid \conf_t] & \ge \expect\left[\max_{i \in [k]} \min_{j \neq i} \left\{\Pr(\WW_i^{(t+1)}) - \Pr(\WW_j^{(t+1)})\right\} \st \conf_t\right] \nonumber \\
    & \ge \expect\left[\min_{j \neq 1} \left\{\Pr(\WW_1^{(t+1)}) - \Pr(\WW_j^{(t+1)})\right\} \st \conf_t\right] \nonumber \\
    & \ge \expect\left[\Pr(\WW_1^{(t+1)}) - \Pr(\WW_2^{(t+1)}) \st \conf_t\right]  . \label{eq:prelim:expected-bias}
\end{align}
Suppose that \(p_1^{(t)} = p_2^{(t)} + \delta_t > p_2^{(t)} \ge \ldots \ge p_k^{(t)}\).
This implies that \(\Pr\mspace{2mu}(\WW_1^{(t+1)} \mid \conf_t) - \Pr\mspace{2mu}(\WW_i^{(t+1)} \mid \conf_t) \ge \Pr\mspace{2mu}(\WW_1^{(t+1)} \mid \conf_t) - \Pr\mspace{2mu}(\WW_2^{(t+1)} \mid \conf_t)\) for all \(i \ge 3\) since \(p_2^{(t)} \ge p_i^{(t)}\) for \(i \ge 3\).
If we proved that \(\Pr\mspace{2mu}(\WW_1^{(t+1)} \mid \conf_t) - \Pr\mspace{2mu}(\WW_2^{(t+1)} \mid \conf_t) \ge \delta_t (1+x)\) for some constant \(x > 0\), then we would get \(\expect[\delta_{t+1} \mid \conf_t] \ge \delta_t(1 + x) \), which allows us to concentrate and establish that \(\delta_{t+1} \ge \delta_t(1 +x/2)\) w.h.p.\ (conditional on \(\conf_t\)), for a suitable choice of \(x > 0\) and for a large enough \(\delta_t\), thanks to the Hoeffding bound (\cref{lemma:app:hoeffding}).

Notice that, in general, the probability of \(\WW_i\) is the probability that \(X_i\) is either the unique maximum or, if it is a maximum and is not unique, it must win the tie broken u.a.r.
In formulas, the probability is
\begin{align}
    \Pr(\WW_i) = \sum_{m = 1}^k \frac{1}{m} \sum_{\substack{1 \le i_1 < \ldots < i_m \le k: \\ i_j \neq i \, \forall j\in[m]}} \Pr(\bigcap_{j\in[m]} \{X_i = X_{i_j}\} \cap (\bigcap_{j \notin \{i,i_1,\ldots,i_m\}} \{X_i > X_j\})). \label{eq:prelim:prob-winning}
\end{align}
Using directly \cref{eq:prelim:prob-winning}  to bound \(\Pr(\WW_1) - \Pr(\WW_2)\) is difficult, and we did not find general estimations in the literature.

\paragraph*{\texorpdfstring{Trivial case: very large \(h\)}{Trivial case: very large h}.}
If, e.g., \(h\) is very large compared to \(\delta\), in one round all nodes would adopt the plurality opinion with high probability.
The reason follows.

Suppose \(p_1 = p_2 + \delta > p_2 \ge \ldots \ge p_k\).
Then, by the multiplicative Chernoff bound (\cref{lemma:app:multiplicative-chernoff} in \cref{sec:tools}) we get that
\[
    \Pr(X_1 \ge hp_1 (1 - x)) \ge 1 - \exp(- \Theta(x^2 hp_1)).
\]
Similarly, 
\[
    \Pr(X_i \le hp_2 (1 + x)) \ge 1 - \exp(- \Theta(x^2 hp_2)).
\]
Suppose \(p_1\) and \(p_2\) are comparable, that is, \(p_1 = \Theta(p_2)\) (anyway the bias we require is always an \(o(p_1)\)).
Then, \(X_1 \ge hp_1 (1 - x)\) and \(X_i \le hp_2 (1 + x)\) for all \(i \ge 2\) hold w.h.p.\ as long as \(h \ge C \log n / (p_1 x^2)\) for a large enough constant \(C > 0\).
If \(hp_1(1-x) > hp_2(1+x)\), then \(\WW_1\) holds w.h.p., that is, a node would adopt opinion 1 in just 1 round w.h.p.
Playing with constants, by the union bound, one would obtain plurality consensus w.h.p.\ in just 1 round.
The condition is satisfied whenever \(\delta > x(p_1 + p_2)\), that is, for 
\(x < \delta / (p_1 + p_2)\).
Since \(h \ge C \log n / (p_1 x^2)\) and \(x < \delta / (p_1 + p_2)\) must hold at the same time, it is sufficient to choose \(h \ge C' p_1 \log n / \delta^2\) for a large enough constant \(C' > C\).
Notice that, since our bias can be such that \(\delta \sim \sqrt{p_1/n}\), we must have that \(hp_1 > (C' / \lambda_1^2)p_1 \cdot n \log n \gg \log n\) for any choice of \(p_1 = \Omega(\log n / n)\).
The minimum that \(h\) can be with this approach is \(\Theta(\log^2 n / p_1)\) and can be as large as  \(h = \Omega(n^{2/3} / p_1)\) if \(p_1 \ge 1/n^{1/3}\).
Our \cref{thm:prelim:main} guarantees us that we can always set \(h = \Theta(\log n / p_1)\) as long as \(p_1 \ge \lambda_2 \log n / n\) for a large enough constant \(\lambda_2 > 0\).
For the minimum bias \(\delta = \lambda_1 \sqrt{p_1/n}\), we have that \(p_2 = p_1 - o(p_1)\).
Hence, \(\expect[X_2] = hp_2 = hp_1 - o(hp_1) = \expect[X_1] - o(\log n)\). 
Unfortunately, we cannot capture such a deviation from the average of $X_1$ using concentration bounds, so we need to adopt a different approach.

\subsection{Our approach.}
\label{sec:our_approach}

At the heart of our analysis, there is a technical lemma that was proved in \cite[Lemma 9]{FraigniaudN19}.

The lemma can be reformulated as follows.
Suppose we are performing the \hmaj when in the system there are only \(k=2\) opinions.
Then, \cite[Lemma 9]{FraigniaudN19} provides a tight lower bound on the expected bias after one round, i.e., on $\pr{\WW_1}-\pr{\WW_2}$.
Notice that, when \(k = 2\), $\pr{\WW_1}-\pr{\WW_2} = \Pr(X_1 > X_2) + \Pr(X_1 = X_2)/2 - \Pr(X_2 > X_1) - \Pr(X_1 = X_2)/2$, which is equal to \(\Pr(X_1 > X_2) - \Pr(X_2 > X_1)\).

\begin{lemma}[Lemma 9 in \cite{FraigniaudN19}]
\label{lemma:from_lemma_9}
    For any integer \(h > 0\), let $X_1 \sim \binomial(h,p)$ and $X_2=h-X_1 \sim \binomial(h,1 - p)$, with \(p > 1/2\). 
    Let $\delta = 2p - 1$ be the bias. 
    Then, we have that
    \begin{align*}
        \Pr\pa{ X_1 > X_2} - \Pr\pa{ X_2 > X_1} \geq \sqrt{\frac{2h}{\pi}} \cdot  g(\delta, h),
    \end{align*}
    where
    \begin{equation*}
        g(\delta, h)=
        \begin{cases}
            \delta \, \pa{1-\delta^2}^{\frac{h-1}{2}} & \text{if } \delta < \frac{1}{\sqrt{h}}, \\
            \frac{1}{\sqrt{h}} \, \pa{1-\frac{1}{h}}^{\frac{h-1}{2}} & \text{if } \delta \geq \frac{1}{\sqrt{h}}.
        \end{cases}
    \end{equation*}
\end{lemma}

We show how to adapt the proof of \cite[Lemma 9]{FraigniaudN19} to get \cref{lemma:from_lemma_9} in \cref{sec:omitted:preliminaries}.
Basically, \cref{lemma:from_lemma_9} makes explicit the minimum sample size required by the expectation of the bias in the next round to increase by a constant multiplicative factor relative to the current bias, to which we refer here as $\delta$. 
More specifically, it is sufficient that the number of samples $h$ satisfies $\sqrt{\frac{2h}{\pi}}\pa{1-\delta^2}^{\frac{h-1}{2}}\geq \sqrt{\frac{2h}{\pi \,e}}\geq c$ for some positive constant \(c > 0\) as long as $\delta<\frac{1}{\sqrt{h}}$. 
If, instead, $\delta\geq\frac{1}{\sqrt{h}}$, the lemma states that the new expected value of the bias is more than a constant, which allows us to show that the bias increases just by standard concentration arguments around the averages of \(X_1\) and \(X_2\) (through Berry-Esseen's inequality or Chernoff bounds depending on how large \(\delta\) is w.r.t.\ \(h\)).

We would like to exploit something that is similar, in spirit, to \cref{lemma:from_lemma_9} when \(k \gg 1\), in order to get that the bias increases each round.
However, \(X_1 \sim \text{Binomial}(h,p_1)\) and \(X_2 \sim \text{Binomial}(h,p_2)\) but \(p_1 + p_2 < 1\) and \(X_2 \neq h - X_1\).
In order to overcome this issue, let us define few events.

\begin{definition}[Winning events]
    \label{def:W_1-2_definition}
    Let \(i \in [k]\).
    We define $\winties{i}$ to be the event in which opinion "$i$" is the most sampled opinion, including possible ties. 
    Formally,
    \begin{equation}
        \winties{i}=\cap_{j \in [k]\setminus \{i\}} \cpa{X_i \geq X_j}.
    \end{equation}
    Similarly, we define $\winstrict{i}$ to be the event in which opinion "$i$" is the most sampled opinion, {without} ties. 
    Formally,
    \begin{equation}
        \winstrict{i}=\cap_{j \in [k]\setminus \{i\}} \cpa{X_i > X_j}.
    \end{equation}
    Finally, for \(i \neq j \in [k]\), we define \(\WW_{i,j} = \WW_i \cup \WW_j\), $\winstrict{i,j} = \winstrict{i} \cup \WW_{j,\text{strict}}$, and \(\winties{i,j} = \winties{i} \cup \WW_{j,\text{ties}}\).
\end{definition}

Clearly, \(\winstrict{i} \subseteq \WW_i \subseteq \winties{i}\) and \(\winstrict{i,j} \subseteq \WW_{i,j} \subseteq \winties{i,j}\).
The idea is to perform a conditional analysis: if we condition on \(\WW_{1,2}\), then we already know that either \(X_1\) or \(X_2\) win at the next round.
However, we do not quite get that the conditional random variable \((X_1 \mid \WW_{1,2})\) is a binomial one.
At the same time, it is not true that \((X_2 \mid \WW_{1,2})\) counts the number of failures of \((X_1 \mid \WW_{1,2})\).
The whole \cref{subsection:reduce2opinions} presents the analytical effort to demonstrate that, in practice, we can get some result that is similar to \cref{lemma:from_lemma_9}. 
The section is quite technical and involved, with ad-hoc results that are needed for this adaptation.
We stress that we actually use, as conditional event, the event \(\winstrict{1,2}\).
The reason being that (as we show in \cref{lemma:difference_maj_is_difference_comparison})
\begin{align}
    & \pr{\WW_1 \st \winstrict{1,2}} - \pr{\WW_2 \st \winstrict{1,2}} \nonumber\\
    = \ & \pr{X_1 > X_2 \st \winstrict{1,2}} - \pr{X_2 > X_1 \st \winstrict{1,2}}, \label{eq:prelim:nice-diff}
\end{align}
while the equality is not true if instead of \(\winstrict{1,2}\) we use \(\WW_{1,2}\), because we need to take ties into account.
Observe that \cref{eq:prelim:nice-diff} gives a formula that resembles the one that is estimated in \cref{lemma:from_lemma_9}.

In \cref{sec:analysis:ties-estimation}, we estimate the probability of \(\winstrict{1,2}\) with respect to \(\WW_{1,2}\) and \(\winties{1,2}\).
In order to do that, we classify the \(k\) opinions in two classes: the \textit{strong} and the \textit{rare} opinions.
The strong ones are opinions whose probabilities are comparable to \(p_1\), say, \(p_i > p_1 / 2\).
The rare ones are all the others.
Through concentration arguments, we can show that rare opinions disappear in one round w.h.p., and it is easy to see, by symmetry arguments, that \(\WW_1 = \Omega(p_1)\).
The main result that we obtain in \cref{sec:analysis:ties-estimation} is that, basically, under the hypothesis that \(h p_1 \ge C \log n\) for some large enough constant \(C > 0\), \(p_1 / 8 \le \pr{\winties{1}}/8 \le \Pr(\winstrict{1}) \le \pr{\WW_1} \le \pr{\winties{1}}\), that is,  \textit{ties do not matter that much}.
Notice that the latter fact is, again, trivial if \(h p_1 = \Omega(\max\{p_1^2 n, \log^2 n\} \gg O(\log n)\), while we require a smaller (and, possibly, \textit{much} smaller) \(h\).
Also, observe that \( \Pr(\winstrict{1}) \ge p_1/8 \) implies that \(\Pr(\winstrict{1,2}) = \Omega(p_1 + p_2)\).
Computing directly the probability of ties in the multinomial is hard, but we resolve this issue mapping injectively events in which opinion 1 wins with a tie to an event in which opinion 1 wins without a tie.
We show that this mapping preserves probabilities up to small constants, therefore we can conclude that events with ties (where opinion 1 wins) have probability weights that are comparable to events without ties (where opinion 1 wins). 
See \cref{lemma:relationship_ties_without_ties} for more details (and note that no initial bias is required for the result of \cref{lemma:relationship_ties_without_ties} to hold).

Finally, in \cref{sec:analysis:alltogether} we show how to combine all the ingredients that we have proved in \cref{subsection:reduce2opinions,sec:analysis:ties-estimation} to obtain the round-by-round expected growth of the bias and the final statement on plurality consensus.
As for the round-by-round expected growth of the bias, we just use the law of total probability which implies that 
\begin{align}
    & \pr{\WW_1} - \pr{\WW_2} \\
    \ge \ & \pr{\winstrict{1,2}} \left[\pr{\WW_1 \st \winstrict{1,2}} - \pr{\WW_2 \st \winstrict{1,2}}\right] \nonumber \\
    = \ & \pr{\winstrict{1,2}} \left[\pr{X_1 > X_2 \st \winstrict{1,2}} - \pr{X_2 > X_1 \st \winstrict{1,2}}\right],
\end{align}
and then use \cref{eq:prelim:expected-bias}.
Given the expected growth of the bias, we make use of the Bernstein's inequality (\cref{lemma:bernstein_inequality} in \cref{sec:tools}) to show its round-by-round increase in high probability, and we need to make sure that all other surrounding conditions of the growth are still satisfied in order to iterate in high probability each round.

\subsection{Final discussion and open questions}\label{sec:final-discussion}

In this paper we present an analysis of the \hmaj dynamics when \(h p_1 \ge C \log n\) for a large enough constant \(C\).
Interestingly, \cref{thm:prelim:main} answers (partially) a long-standing open question on mathoverflow \cite{210018} which basically asks how to estimate \(\Pr(\WW_1) - \Pr(\WW_2)\) when \(h = o(1/\delta^2)\), that is, when \(\delta = o(\sqrt{1/h})\).
With \cref{proposition:bias_hp>1} we answer to this question under the hypothesis that \(h p_1 \ge C \log n\), and we also show that the bias increases w.h.p.\ each round as long as \(C' \sqrt{p_1 / n} \le \delta = o(\sqrt{p_1 / \log n})\) for some large enough constant \(C' > 0\).

We also remark that \(\delta \sim \sqrt{p_1 / n}\) is the smallest bias that was ever required in the literature on opinion dynamics (see \cref{sec:related-works}). 
Indeed, in the proof of \cref{thm:hp_1 large}, we show that the standard deviation of the bias at the next round is no smaller than some function \(\Theta(\sqrt{p_1/n})\), while its expected growth is at least a multiplicative factor \(\sim \sqrt{\log n}\).
It remains open to understand whether the bias can be reduced for \(h \sim k \log n\), which would imply that the growth of the bias is higher than a multiplicative factor \(\sim \sqrt{\log n}\): however, we conjecture that this growth is optimal in the regime \(h \sim k \log n\).

The main open question here is whether the lower bound \(\Omega(k/ h^2)\) given by \cite{becchettiSimpleDynamicsPlurality2017} is tight or not. 
To understand this, one has to improve the assumption on \(hp_1\) and analyze especially the case where \(hp_1 \ll 1\).
We argue more in the next subsection.

We remark that our work does not deal with perfectly-balanced configurations: one simply has to show that in short time the system reaches a configuration with the bias specified in the hypotheses of \cref{thm:prelim:main}.
We leave the analysis of the symmetry-breaking phase for future research.

\paragraph*{On the multinomial distribution.}

Our analysis offers a corollary for the multinomial distribution that is interesting per-se, and is just a reformulation of \cref{lemma:relationship_ties_without_ties}.

\begin{corollary}\label{cor:multinomial}
    Consider a multinomial random variable \((X_1, \ldots, X_k)\) where each \(X_i \sim \binomial(h, p_i)\), with \(p_1 \ge \ldots \ge p_k\).
    If \(h p_1 \ge C \log n\) for a large enough constant \(C > 0\), then \(\Pr(\cap_{i = 2}^k \{X_1 > X_i\}) \ge c p_1\) for some small enough constant \(c > 0\). 
\end{corollary}


We do not know how the statement of \cref{cor:multinomial} changes when \(h p_1\) decreases, which remains an open question.
This is linked to the generalized birthday paradox, which asks as follows:
If \(p_1 = \ldots = p_k\), what is the probability that there exists \(i \in [k]\) such that \(X_i > 1\)?
In fact, if \(h p^2_1 = \Theta(1)\), we are in the standard birthday paradox:
with constant probability, we have that \(\max_{i\in[k]} \cpa{X_i} = 1\), which means that whenever opinion 1 is the maximum, the number of ties is $h$ with constant probability. In contrast, \cref{cor:multinomial} suggests that when \(hp_1 = \Omega(\log n)\) the number of ties that involve opinion 1 when it is the maximum is small.

Since the number of ties among the maxima increases when \(hp_1\) gets smaller, we conjecture that, for \(hp_1 \ll 1\), \(\Pr(\winstrict{1}) = o(p_1)\), implying that our approach, that entirely builds around the fact that \(\Pr(\winstrict{1}) \sim \Pr(\WW_1)\), cannot apply, and new ideas are required.

Another open question is to quantify the number of ties among the maxima, in general, and especially when opinion \(1\) is in the maxima for a generic configuration where \(p_1 \ge \ldots \ge p_k\). 
This would help bridging the two cases \(hp_1 \ll 1\) and \(hp_1 \gg 1\).

\section{Analysis}
In this section we present the proof of \Cref{thm:prelim:main}, following the overview presented in \Cref{sec:intro:sketch}. 
Most of the technical proofs are deferred to \Cref{sec:missing-proofs}.

\subsection{\texorpdfstring{Conditional analysis of the expected behavior of the process}{Conditional analysis of the expected behavior of the process}}
\label{subsection:reduce2opinions}
Due to space limitations, we defer all the proofs in this section to \Cref{sec:omitted:reduce2opinions}.

\Cref{lemma:from_lemma_9} shows a good lower bound of $\pr{\WW_1}-\pr{\WW_2}$ for the \hmaj when the number of opinion is 2. In order to use this result for $k$ opinions, our plan is to condition on the event that either opinion 1  or 2 wins, without ties. Let $x=\max_{i\geq 3}\cpa{X_i}$ the mode of the sample among the opinions $\cpa{3,\ldots, k}$. This event can be written as $\max\cpa{X_1,X_2}>x$.
In the following lemma, we adapt the result of \Cref{lemma:from_lemma_9} to work under the additional conditioning event that opinion 1 or 2 is sampled more than $x$ times.

\begin{lemma} \label{lemma: reduction to 2 opinions}
    For any integer \(m > 0\), let $Y_1\sim \binomial(m,q)$ and $Y_2=m-Y_1\sim \binomial(m,1 - q)$, with $q > 1/2$. We have, for any $m/2 < x \leq m$,
    \[
        \Pr\left( Y_1>Y_2 \mid \max(Y_1,Y_2)>x \right) - \Pr\left( Y_2>Y_1 \mid \max(Y_1,Y_2)>x \right) \geq C_1 \cdot \min(\sqrt{m}\cdot (2q-1),1),
    \]
    for some constant $C_1>0$.
\end{lemma}

The next lemma shows that the expected bias at the next round, i.e. $\pr{\WW_1}-\pr{\WW_2}$, is equal to $\pr{X_1>X_2}-\pr{X_2>X_1}$, conditioning on the event either opinion 1 or 2 wins, without ties. The last quantity is easier to estimate, because we don't have to handle possible ties. 
This is to get a formula that resembles that of \Cref{lemma:from_lemma_9}.

\begin{lemma} Recall the \Cref{def:W_1-2_definition} of the event $\winstrict{1,2}$. We have
    \label{lemma:difference_maj_is_difference_comparison}
    \begin{align*}
        & \ \Pr\left( \WW_{1} \mid \winstrict{1,2} \right) - \Pr\left( \WW_{2} \mid \winstrict{1,2} \right)
        \\
        = & \ \Pr\left( X_1 > X_2 \mid \winstrict{1,2} \right) - \Pr\left( X_2 > X_1 \mid \winstrict{1,2} \right)
    \end{align*}
\end{lemma}

We remark that the statement of \Cref{lemma:difference_maj_is_difference_comparison} does not hold if we condition on the event $\WW_{1,2}$, which includes possible ties. In fact, the event $\cpa{\WW_1,\WW_{1,2}}$ contains sub-events for which opinion 1 ties with many other opinions, and we could not find a simple close expression for the probability of these sub-events.

To use \Cref{lemma: reduction to 2 opinions} after conditioning on the fact either opinion 1 or 2 wins, we need to estimate the number of total samples of opinion 1 and 2. Since, at the end, we assume $hp_1 \ge \log n$, we can apply the Chernoff bound (\Cref{lemma:app:multiplicative-chernoff} in \Cref{sec:tools}) to show that $X_1+X_2\geq (p_1+p_2)/2$ w.h.p., but we must ensure that this remains true under the aforementioned conditioning. The following lemma shows this is indeed the case. Despite the intuitive nature of the statement (the fact either opinion 1 or 2 wins is positively correlated with the absolute value of $X_1$ and $X_2$) we need to use an involved coupling argument to formally prove it.
\begin{lemma} \label{lemma:sum_12_is_independent_maximum}
                It holds that
                \begin{equation*} 
                    \Pr\left( X_1 + X_2 \geq h\cdot (p_1+p_2)/2 \mid \winstrict{1,2} \right)\geq \Pr\left( X_1 + X_2 \geq h\cdot (p_1+p_2)/2 \right).
                \end{equation*}
\end{lemma}

\subsection{Estimating the probability of the conditional event}\label{sec:analysis:ties-estimation}
In this section our goal is to estimate a lower bound for \(\pr{\winstrict{1,2}}\).
More specifically, we show that $\pr{\winstrict{1,2}}\geq C (p_1+p_2)$ for some constant $C>0$. 
As we already argued after \Cref{lemma:difference_maj_is_difference_comparison}, it is important to exclude ties in the event we condition on.

The next definitions introduce terms to refer to opinions that will disappear w.h.p.\ at the next round, the \textit{weak opinions}, and their counter part, the \textit{strong opinions}.
\begin{definition}[Rare opinions]
    Denote with $\mathcal{R}_x:= \cpa{i\in [k] : p_i\leq x\cdot p_1 }$ be the set of the $x$-\textit{rare opinions}, for $0<x<1$.
\end{definition}
\begin{definition}[Strong opinions]
    Denote with $\mathcal{S}:= \cpa{i\in [k] : p_i > \frac{1}{2} p_1 }$ be the set of the \textit{strong opinions}, for $0<x<1$. 
\end{definition}
Note that the set of the strong opinions and the set of $1/2$-rare opinions are complementary.

In the next lemma we show that w.h.p.\ all weak opinions will be sampled less than the opinion 1, and therefore, they will disappear at the next round.
The proof is just an application of Chernoff bounds (\Cref{lemma:app:multiplicative-chernoff} in \Cref{sec:tools}) and we defer it to \Cref{sec:omitted:ties-estimation}.
\begin{lemma}
\label{claim:weak_opinion_lose_whp}
    Let $C_2, C_3$ be two constants s.t. $0<C_2<1$ and $C_3>2$. If $h p_1 > C_4 \log n$, with $C_4= \pa{3\, C_3\frac{1}{1-C_2}}^2$, then
    \[
        \pr{\cap_{i\in\mathcal{R}_{C_2}} \cpa{X_1 > X_i}} \geq 1-\frac{1}{n^{C_3-2}}.
    \]
\end{lemma}

To bound $\pr{\WW_{1,\text{strict}}}$, we plan to first show that there exists a constant $C$ s.t. $\winties{1}\geq C p_1$, and, therefore, show that $\pr{\WW_{1,\text{strict}}}\geq C \Pr(\winties{1})$, for some other constant $C$. Hence, in the next lemma we show that $\Pr(\WW_{1})\geq C p_1$ which, in turn, implies that \(\Pr(\winties{1}) \ge Cp_1\). The idea is to use \Cref{claim:weak_opinion_lose_whp} to show that only strong opinion compete for the win and in the worst case they have all same probability to win without ties. The thesis follows after showing that the number of strong opinions is at most $2/p_1$, and that $\sum_{i\in\mathcal{S}}\pr{\WW_i}\geq 1$.
Again, the proof is deferred to \Cref{sec:omitted:ties-estimation}.
\begin{lemma}
    \label{lemma:bound_W_1-2}
    Let $C_4=\pa{18}^2$. If $p_1 h \geq  C_4\log n$, then
    \[
        \Pr(\WW_{1}) \geq \frac{1}{3} p_1.
    \]
\end{lemma}

The next Lemma plays a key role in avoiding the computation of the probability of ties. Assuming that $hp_1 = O(\log n)$, we can map each realization of the multinomial in which opinion 1 wins with ties to a realization in which opinion 1 wins without ties. We will show that these two realizations have the same probability up to a constant. Therefore, we can conclude that tie events have a comparable weight to events without ties.
We present its proof in \cref{sec:omitted:ties-estimation} as we believe this is an interesting result per-se, which also proves \Cref{cor:multinomial} in \Cref{sec:final-discussion}.
\begin{lemma}
\label{lemma:relationship_ties_without_ties}
    Let $hp_1 > C_4 \log n$ for $C_4=18^2$. It holds that
    \[
        \pr{\WW_{1,\text{strict}}}   \geq \frac{1}{6}\pr{\winties{1}} \geq \frac{1}{18} p_1.
    \]
\end{lemma}

The goal of this section was to show that $\pr{\WW_{1,2,\text{strict}}}\geq C \pa{p_1+p_2}$ for some constant $C>0$. This can be easily obtained by applying \Cref{lemma:relationship_ties_without_ties}, because $p_1+p_2 \leq 2p_1$ and $\WW_{1,\text{strict}}\subseteq \WW_{1,2,\text{strict}}$.
\begin{corollary}
    \label{lemma:bound_M_1-2}
    Let $hp_1 > C_4 \log n$ for $C_4=18^2$. It holds that
    \[
        \pr{\WW_{1,2,\text{strict}}} \geq \frac{1}{36} \pa{p_1+p_2}.
    \]
\end{corollary}

\subsection{Putting everything together}\label{sec:analysis:alltogether}

In this section we can finally put together all the results to give a lower bound to $\Pr\left( \WW_{1} \right) - \Pr\left( \WW_{j} \right)$ and to use it to prove \Cref{thm:prelim:main}.
Next lemma is the statement that is ``equivalent in spirit'' to \Cref{lemma:from_lemma_9} but for many opinions, conditional on \(\winstrict{1,2}\).
Its proof involves a mix of the previous results and is quite intricate.
\begin{lemma}
    \label{lemma:lemma_ema_with_conditioning}
    Let $p_1\,h\geq C_4\,\log n$, with $C_4=18^2$. We have
    \[
        \Pr\left( \WW_{1} \mid \WW_{1,2,\text{strict}} \right) - \Pr\left( \WW_{2} \mid \WW_{1,2,\text{strict}} \right) \geq C \min\left\{\frac{ (p_1-p_2) \sqrt{h}}{\sqrt{2(p_1+p_2)}},1\right\},
    \]
    for some constant $C>0$.
\end{lemma}

\begin{proof}[Proof of \Cref{lemma:lemma_ema_with_conditioning}]
    Let
    \[
        M= \{m\in \mathbb{N} :  \Pr\left( X_1 + X_2 = m \mid \winstrict{1,2} \right)>0\} \cap \{m \geq h \cdot (p_1 + p_2)/2 \}
    \]
    and 
    \[
    L_m= \{x\in \mathbb{N} :  \Pr\left( \max_{ j \geq 3 } X_j = x \mid   X_1 + X_2 = m , \winstrict{1,2} \right)>0\}.
    \]
    We have
    \begin{align}
        \label{eq:many_law_of_total_prob}
        \nonumber
        &\Pr\left( \WW_{1} \mid \winstrict{1,2} \right) - \Pr\left( \WW_{2} \mid \winstrict{1,2} \right) \\ 
        \nonumber
        \underset{(i)}{=} \ & \Pr\left( X_1 > X_2 \mid \winstrict{1,2} \right) - \Pr\left( X_2 > X_1 \mid \winstrict{1,2} \right)\\ 
        \nonumber
        \underset{(ii)}{\geq} \ & \sum_{ m \in M} \Pr\left( X_1 + X_2 = m \mid \winstrict{1,2} \right) \left[ \Pr\left( X_1 > X_2 \mid \winstrict{1,2} , X_1 + X_2 = m \right) \right. 
        \\
        \nonumber
        & - \left. \Pr\left( X_2 > X_1 \mid \winstrict{1,2}, X_1 + X_2 = m \right) \right]  \\ 
        \nonumber
        \underset{(iii)}{=} \ & \sum_{ m \in M} \Pr\left( X_1 + X_2 = m \mid \winstrict{1,2} \right) \\
        \nonumber
        & \cdot \sum_{x \in L_m} \Pr\left( \max_{ j \geq 3 } X_j = x \mid   X_1 + X_2 = m , \winstrict{1,2} \right)\\  
        \nonumber
        & \cdot \left[ \Pr\left( X_1 > X_2 \mid \winstrict{1,2} , X_1 + X_2 = m, \max_{ j \geq 3 } X_j = x \right)  \right.
        \\
        \nonumber
        & - \left. \Pr\left( X_2 > X_1 \mid \winstrict{1,2}, X_1 + X_2 = m , \max_{ j \geq 3 } X_j = x \right) \right] \\ 
        \nonumber
        \underset{(iv)}{=} \ & \sum_{ m \in M} \Pr\left( X_1 + X_2 = m \mid \winstrict{1,2} \right) 
        \\
        \nonumber
        & \cdot \sum_{x \in L_m} \Pr\left( \max_{ j \geq 3 } X_j = x \mid   X_1 + X_2 = m , \winstrict{1,2} \right)\\  
        \nonumber
        & \cdot \left[ \Pr\left( X_1 > X_2 \mid \max\{X_1, X_2\} > x , X_1 + X_2 = m, \max_{ j \geq 3 } X_j = x \right) \right. \\
        & - \left. \Pr\left( X_2 > X_1 \mid \max\{X_1, X_2\}> x, X_1 + X_2 = m , \max_{ j \geq 3 } X_j = x \right) \right],
     \end{align}
     where in $(i)$ we used \Cref{lemma:difference_maj_is_difference_comparison}, in $(ii),(iii)$ we used the law of total probabilities, and in $(iv)$ we used that, by \Cref{def:W_1-2_definition}, 
     \begin{align*}
                    & \ \left( \winstrict{1,2},  X_1 + X_2 = m , \max_{ j \geq 3 } X_j = x \right) 
                    \\
                    = & \  \left( \max\{X_1, X_2\}> x , X_1 + X_2 = m, \max_{ j \geq 3 } X_j = x \right).
    \end{align*}
     Conditioning on $\{X_1 + X_2 = m\}$, makes $X_1, X_2$ independent from $X_3,\ldots, X_k$, and distributed as $Y_1,Y_2$, where $Y_1 \sim \binomial(m,q_1)$ with $q_1=\frac{p_1}{p_1+p_2}$ and $Y_2=m-Y_1$. We obtain, by \Cref{eq:many_law_of_total_prob}
     \begin{align}
        \label{eq:final_equation_with_conditioning}
        \nonumber
        &\Pr\left( \WW_{1} \mid \winstrict{1,2} \right) - \Pr\left( \WW_{2} \mid \winstrict{1,2} \right)
        \\
        \nonumber
         \geq \ &\sum_{ m \in M} \Pr\left( X_1 + X_2 = m \mid \winstrict{1,2} \right) 
         \\
         \nonumber
         & \cdot \sum_{x \in L_m} \Pr\left( \max_{ j \geq 3 } X_j = x \mid   X_1 + X_2 = m , \winstrict{1,2} \right)
         \\
         \nonumber
          & \cdot \left[ \Pr\left( Y_1 > Y_2 \mid \max\{Y_1, Y_2\}\geq x \right) - \Pr\left( Y_2 > Y_1 \mid  \max\{Y_1, Y_2\}\geq x \right) \right]
         \\
         \nonumber
         \underset{(i)}{\geq} \  & \sum_{ m \in M} \Pr\left( X_1 + X_2 = m \mid \winstrict{1,2} \right) 
         \\
         \nonumber
         & \cdot \sum_{x \in L_m} \Pr\left( \max_{ j \geq 3 } X_j = x \mid   X_1 + X_2 = m , \winstrict{1,2} \right)\\
         \nonumber
          & \cdot C_1 \cdot \min\pa{\sqrt{m} \pa{\frac{\delta_2}{p_1+p_2}},1}\\
         \nonumber
         \underset{(ii)}{=} \ & \sum_{ m \in M} \Pr\left( X_1 + X_2 = m \mid \winstrict{1,2} \right) \cdot C_1 \cdot \min\pa{\sqrt{m} \pa{\frac{\delta_2}{p_1+p_2}},1}\\
         \nonumber
         \underset{(iii)}{\geq} \ &  \pr{ X_1 + X_2 \geq h \cdot (p_1 + p_2)/2 \mid \winstrict{1,2}} \cdot C_1 \cdot \min\pa{ \pa{\frac{\delta_2 \sqrt{h}}{\sqrt{2(p_1+p_2)}}},1} \\
         \nonumber
         \underset{(iv)}{\geq} \ & C_1 \pa{1-\exp\pa{-\frac{1}{8}h(p_1+p_2)}}  \cdot \min\pa{ \pa{\frac{\delta_2 \sqrt{h}}{\sqrt{2(p_1+p_2)}}},1}
         \\
         \nonumber
         \underset{(v)}{\geq} \ & C_1 \pa{1-\frac{1}{n^2}}  \cdot \min\pa{ \pa{\frac{\delta_2 \sqrt{h}}{\sqrt{2(p_1+p_2)}}},1}
         \\
         \nonumber
         \geq \ & C \min\pa{ \pa{\frac{\delta_2 \sqrt{h}}{\sqrt{2(p_1+p_2)}}},1},
    \end{align}
    where in $(i)$ we used \Cref{lemma: reduction to 2 opinions}, in $(ii)$ we used the law of total probabilities, in $(iii)$ we used that by definition of $M$, $m\geq h(p_1+p_2)/2$, in $(iv)$ we used \Cref{lemma:sum_12_is_independent_maximum} and the multiplicative Chernoff bound (\Cref{lemma:app:multiplicative-chernoff} in \Cref{sec:tools}), and in $(v)$ we used that $hp_1 \geq 16 \log n$. This concludes the proof of \Cref{lemma:lemma_ema_with_conditioning}.
\end{proof}

In the next lemma, we finally remove the conditional event \(\winstrict{1,2}\) and bound $\pr{\WW_1}-\pr{\WW_2}$ from below. 
We need to use \Cref{lemma:relationship_ties_without_ties,lemma:lemma_ema_with_conditioning}, but its proof consists in just mixing in the ``right way'' previous results, hence we defer it to \Cref{sec:omitted:alltogether}.
From now on, we denote the difference \(p_1 - p_j\) by \(\delta(j)\).

\begin{lemma}
    \label{proposition:bias_hp>1} 
    Let $p_1\,h\geq C_4\,\log n$, with $C_4=18^2$. 
    If $\delta(j) \leq \sqrt{\frac{2(p_1+p_2)}{h}}$, then there exists a constant $C_5>0$ s.t.
    \[
        \Pr\left( \WW_{1} \right) - \Pr\left( \WW_{j} \right) \geq C_5 \left (  \delta(j)\sqrt{\frac{h}{p_1+p_2}}  \right )\pr{\WW_1}.
    \]
    Instead, if $\delta(j)\geq \sqrt{\frac{2(p_1+p_2)}{h}}$, there exists a constant $0<C_6<1$ s.t.
    \[
        \Pr\left( \WW_{1} \right) \geq  \frac{1}{1-C_6} \Pr\left( \WW_{j} \right) .
    \]
\end{lemma}

    Given a configuration $\pa{p_1,\dots,p_k}$,
    we denote by \(p_j'\) the probability associated to opinion \(j\) at the next round.
    Fix an opinion $j\neq 1$. Let $\delta(j)'$ be the bias between opinion 1 and opinion j at the next round, i.e. $p_1' - p_j'$.

    Before proceeding in the analysis presentation, let us comment \Cref{proposition:bias_hp>1}. Note that $\Pr\left( \WW_{j} \right)$ is the expected value of $p'_j$. In the case the bias between opinion 1 and opinion $j$ is more than $\sqrt{\frac{2(p_1+p_2)}{h}}$, at the next round we have $p'_1 \geq (1+C) p'_j$ in expectation, for some constant $C>0$. In the case the bias is smaller than $\sqrt{\frac{2(p_1+p_2)}{h}}$, by the fact $h= \Omega\pa{p_1 \log n}$ and $\pr{\WW_1} = \Omega\pa{p_1 + p_2}$, \Cref{proposition:bias_hp>1} implies that at the next round the expected bias grows of a $\Omega\pa{\log n}$ factor. Thanks to standard concentration bounds (Chernoff Bound and Bernstein Inequality) we show that w.h.p. the new bias is sufficiently close to its expectation.

    Now we all have all ingredients to prove \Cref{thm:prelim:main}. Next theorem is just a reformulation of \Cref{thm:prelim:main}.

    \begin{theorem}
    \label{thm:hp_1 large}
    Consider an initial configuration $\pa{p_1,\dots,p_k}$ s.t. $p_1\geq C_7\frac{\log n}{n}$, and $p_1-p_j \geq C_8\sqrt{\frac{p_1}{n}}$ for all $j\geq 2$. Then the \hmaj dynamics  with $h p_1 >  C_4 \log n$ converges in time \(O(\log n)\) to opinion 1, w.h.p., for some constants $C_4, C_7,C_8>0$.
\end{theorem}

    The proof relies on first quantifying the increase in the bias between the first opinion and other opinions after one round. 
    Then it iterates the same reasoning until consensus (for the complete proof, see \Cref{sec:omitted:alltogether}). 
    For the first point, we use concentration bounds on the results obtained in \Cref{proposition:bias_hp>1}. 
    In particular, we distinguish three sub-cases based on the size of the bias between two opinions. 
    For the second point, we need to show that the initial hypotheses are still satisfied after one round to ensure that we can iterate.
    
    In the next lemma, we show that if $p_1 \geq (1+C) p_j$ for some constant $C>0$, opinion $j$ disappears at the next round, by simply applying the Chernoff bound.

\begin{lemma}[Opinion $j$-disappear at the next round]
    \label{lemma:disappear_one_step}
    Let $hp_1 \geq C_4 \log n$.
    If $\delta(j) \geq \pa{1-\frac{1}{1+C_6}}p_1$ for the constant $C_6$ defined in \Cref{proposition:bias_hp>1}, with probability at least $1-\frac{1}{n^4}$ opinion $j$ will disappear at the next round
\end{lemma}

The next Lemma shows that if the bias is $\Omega\pa{\sqrt{p_1/h}}$,  then the new bias will fall within the hypothesis of \Cref{lemma:disappear_one_step} in one step. Therefore, we conclude that, in two steps, opinion $j$ will disappear if the bias is large enough. The proof relies on the results obtained in \Cref{proposition:bias_hp>1}, combined with the Chernoff bound.

\begin{lemma}[Opinion $j$-disappear in two rounds]
    \label{lemma:disappear_two_steps}
    Let $hp_1 \geq C_4 \log n$.
   If $p_1\geq C_7\frac{\log n}{n}$ and $\sqrt{\frac{2(p_1+p_2)}{h}} \leq \delta(j) \leq \pa{1 - \frac{1}{1+C_6}}p_1$ for the constant $C_6>0$ defined in \Cref{proposition:bias_hp>1} and for a constant $C_7>0$, we have that
    \[
        \pr{\delta(j)' \geq \pa{1-\frac{1}{1+C_6}}p_1'} \geq 1 - \frac{2}{n^4}.
    \]
\end{lemma}

The next Lemma addresses the minimum possible bias. We show that, until we fall back into the hypotheses of \Cref{lemma:disappear_one_step,lemma:disappear_two_steps}, the bias grows exponentially. The proof is more complex than the previous two lemmas because it requires bounding the second moment of the bias to apply the Bernstein inequality.

\begin{lemma}[Bias exponential growth]
    \label{lemma:bias_exponential_growth}
    Let $hp_1 \geq C_4 \log n$.
     If $p_1\geq C_7\frac{\log n}{n}$ and $C_8\sqrt{\frac{p_1}{n}} \leq \delta(j) \leq \sqrt{\frac{2(p_1+p_2)}{h}}$ for the constant $C_6>0$ defined in \Cref{proposition:bias_hp>1} and for some constants $C_7,C_8>0$, we have that
    \[
        \pr{\delta(j)' \geq e \, \delta(j)}\geq 1-\frac{1}{n^4}
    \]
\end{lemma}

The next Lemma ensures that the hypothesis of \Cref{thm:prelim:main} on the bias is always satisfied so that we can iterate \Cref{lemma:bias_exponential_growth}.

\begin{lemma}[The new bias satisfies the initial hypothesis]
    \label{lemma:new_bias_hypothesis}
    Let $\delta(j)\geq\sqrt{\frac{p_1}{n}}$, and a constant $C_8>0$. For $n$ large enough we have
    \[
        \pr{\delta(j)' \geq C \sqrt{\frac{p_1'}{n}}}\geq 1-\frac{2}{n^4}.
    \]
\end{lemma}


\section*{Acknowledgments}

This work was partially supported by the MUR (Italy) Department of Excellence 2023 - 2027
for GSSI, and the AID INRIA-DGA project n°2023000872 “BioSwarm”. The authors are grateful to thank Frédéric Giroire for helpful discussions regarding this project.

\bibliography{biblio}

\begin{thebibliography}{10}

\bibitem{AmirABBHKL23}
Talley Amir, James Aspnes, Petra Berenbrink, Felix Biermeier, Christopher Hahn,
  Dominik Kaaser, and John Lazarsfeld.
\newblock Fast convergence of k-opinion undecided state dynamics in the
  population protocol model.
\newblock In Rotem Oshman, Alexandre Nolin, Magn{\'{u}}s~M. Halld{\'{o}}rsson,
  and Alkida Balliu, editors, {\em Proceedings of the 2023 {ACM} Symposium on
  Principles of Distributed Computing, {PODC} 2023, Orlando, FL, USA, June
  19-23, 2023}, pages 13--23. {ACM}, 2023.
\newblock \href {https://doi.org/10.1145/3583668.3594589}
  {\path{doi:10.1145/3583668.3594589}}.

\bibitem{AngluinAE08}
Dana Angluin, James Aspnes, and David Eisenstat.
\newblock A simple population protocol for fast robust approximate majority.
\newblock {\em Distributed Comput.}, 21(2):87--102, 2008.
\newblock URL: \url{https://doi.org/10.1007/s00446-008-0059-z}, \href
  {https://doi.org/10.1007/S00446-008-0059-Z}
  {\path{doi:10.1007/S00446-008-0059-Z}}.

\bibitem{BankhamerBBEHKK22}
Gregor Bankhamer, Petra Berenbrink, Felix Biermeier, Robert Els{\"{a}}sser,
  Hamed Hosseinpour, Dominik Kaaser, and Peter Kling.
\newblock Fast consensus via the unconstrained undecided state dynamics.
\newblock In Joseph~(Seffi) Naor and Niv Buchbinder, editors, {\em Proceedings
  of the 2022 {ACM-SIAM} Symposium on Discrete Algorithms, {SODA} 2022, Virtual
  Conference / Alexandria, VA, USA, January 9 - 12, 2022}, pages 3417--3429.
  {SIAM}, 2022.
\newblock \href {https://doi.org/10.1137/1.9781611977073.135}
  {\path{doi:10.1137/1.9781611977073.135}}.

\bibitem{BecchettiCNPS15}
Luca Becchetti, Andrea Clementi, Emanuele Natale, Francesco Pasquale, and
  Riccardo Silvestri.
\newblock Plurality consensus in the gossip model.
\newblock In Piotr Indyk, editor, {\em Proceedings of the Twenty-Sixth Annual
  {ACM-SIAM} Symposium on Discrete Algorithms, {SODA} 2015, San Diego, CA, USA,
  January 4-6, 2015}, pages 371--390. {SIAM}, 2015.
\newblock \href {https://doi.org/10.1137/1.9781611973730.27}
  {\path{doi:10.1137/1.9781611973730.27}}.

\bibitem{becchettiSimpleDynamicsPlurality2017}
Luca Becchetti, Andrea Clementi, Emanuele Natale, Francesco Pasquale, Riccardo
  Silvestri, and Luca Trevisan.
\newblock Simple dynamics for plurality consensus.
\newblock {\em Distributed Computing}, 30(4):293--306, August 2017.
\newblock \href {https://doi.org/10.1007/s00446-016-0289-4}
  {\path{doi:10.1007/s00446-016-0289-4}}.

\bibitem{BecchettiCNPT16}
Luca Becchetti, Andrea Clementi, Emanuele Natale, Francesco Pasquale, and Luca
  Trevisan.
\newblock Stabilizing consensus with many opinions.
\newblock In Robert Krauthgamer, editor, {\em Proceedings of the Twenty-Seventh
  Annual {ACM-SIAM} Symposium on Discrete Algorithms, {SODA} 2016, Arlington,
  VA, USA, January 10-12, 2016}, pages 620--635. {SIAM}, 2016.
\newblock URL: \url{https://doi.org/10.1137/1.9781611974331.ch46}, \href
  {https://doi.org/10.1137/1.9781611974331.CH46}
  {\path{doi:10.1137/1.9781611974331.CH46}}.

\bibitem{BecchettiCPTVZ24}
Luca Becchetti, Andrea Clementi, Francesco Pasquale, Luca Trevisan, Robin
  Vacus, and Isabella Ziccardi.
\newblock The minority dynamics and the power of synchronicity.
\newblock In David~P. Woodruff, editor, {\em Proceedings of the 2024 {ACM-SIAM}
  Symposium on Discrete Algorithms, {SODA} 2024, Alexandria, VA, USA, January
  7-10, 2024}, pages 4155--4176. {SIAM}, 2024.
\newblock \href {https://doi.org/10.1137/1.9781611977912.144}
  {\path{doi:10.1137/1.9781611977912.144}}.

\bibitem{berenbrink2024}
Petra Berenbrink, Felix Biermeier, and Christopher Hahn.
\newblock Undecided state dynamics with stubborn agents.
\newblock {\em CoRR}, abs/2406.07335, 2024.
\newblock URL: \url{https://doi.org/10.48550/arXiv.2406.07335}, \href
  {https://arxiv.org/abs/2406.07335} {\path{arXiv:2406.07335}}, \href
  {https://doi.org/10.48550/ARXIV.2406.07335}
  {\path{doi:10.48550/ARXIV.2406.07335}}.

\bibitem{BerenbrinkCEKMN17}
Petra Berenbrink, Andrea Clementi, Robert Els{\"{a}}sser, Peter Kling, Frederik
  Mallmann{-}Trenn, and Emanuele Natale.
\newblock Ignore or comply?: On breaking symmetry in consensus.
\newblock In Elad~Michael Schiller and Alexander~A. Schwarzmann, editors, {\em
  Proceedings of the {ACM} Symposium on Principles of Distributed Computing,
  {PODC} 2017, Washington, DC, USA, July 25-27, 2017}, pages 335--344. {ACM},
  2017.
\newblock \href {https://doi.org/10.1145/3087801.3087817}
  {\path{doi:10.1145/3087801.3087817}}.

\bibitem{BerenbrinkCGHKR22}
Petra Berenbrink, Amin Coja{-}Oghlan, Oliver Gebhard, Max Hahn{-}Klimroth,
  Dominik Kaaser, and Malin Rau.
\newblock On the hierarchy of distributed majority protocols.
\newblock In Eshcar Hillel, Roberto Palmieri, and Etienne Rivi{\`{e}}re,
  editors, {\em 26th International Conference on Principles of Distributed
  Systems, {OPODIS} 2022, December 13-15, 2022, Brussels, Belgium}, volume 253
  of {\em LIPIcs}, pages 23:1--23:19. Schloss Dagstuhl - Leibniz-Zentrum
  f{\"{u}}r Informatik, 2022.
\newblock URL: \url{https://doi.org/10.4230/LIPIcs.OPODIS.2022.23}, \href
  {https://doi.org/10.4230/LIPICS.OPODIS.2022.23}
  {\path{doi:10.4230/LIPICS.OPODIS.2022.23}}.

\bibitem{ClementiDGN21}
Andrea Clementi, Francesco D'Amore, George Giakkoupis, and Emanuele Natale.
\newblock Search via {P}arallel {L}{\'{e}}vy {W}alks on {Z2}.
\newblock In Avery Miller, Keren Censor{-}Hillel, and Janne~H. Korhonen,
  editors, {\em {PODC} '21: {ACM} Symposium on Principles of Distributed
  Computing, Virtual Event, Italy, July 26-30, 2021}, pages 81--91. {ACM},
  2021.
\newblock \href {https://doi.org/10.1145/3465084.3467921}
  {\path{doi:10.1145/3465084.3467921}}.

\bibitem{ClementiGGNPS18}
Andrea Clementi, Mohsen Ghaffari, Luciano Gual{\`{a}}, Emanuele Natale,
  Francesco Pasquale, and Giacomo Scornavacca.
\newblock A tight analysis of the parallel undecided-state dynamics with two
  colors.
\newblock In Igor Potapov, Paul~G. Spirakis, and James Worrell, editors, {\em
  43rd International Symposium on Mathematical Foundations of Computer Science,
  {MFCS} 2018, August 27-31, 2018, Liverpool, {UK}}, volume 117 of {\em
  LIPIcs}, pages 28:1--28:15. Schloss Dagstuhl - Leibniz-Zentrum f{\"{u}}r
  Informatik, 2018.
\newblock URL: \url{https://doi.org/10.4230/LIPIcs.MFCS.2018.28}, \href
  {https://doi.org/10.4230/LIPICS.MFCS.2018.28}
  {\path{doi:10.4230/LIPICS.MFCS.2018.28}}.

\bibitem{CooperMRSS25}
Colin Cooper, Frederik Mallmann{-}Trenn, Tomasz Radzik, Nobutaka Shimizu, and
  Takeharu Shiraga.
\newblock Asynchronous 3-majority dynamics with many opinions.
\newblock In Yossi Azar and Debmalya Panigrahi, editors, {\em Proceedings of
  the 2025 Annual {ACM-SIAM} Symposium on Discrete Algorithms, {SODA} 2025, New
  Orleans, LA, USA, January 12-15, 2025}, pages 4095--4131. {SIAM}, 2025.
\newblock \href {https://doi.org/10.1137/1.9781611978322.140}
  {\path{doi:10.1137/1.9781611978322.140}}.

\bibitem{CrucianiMQR23}
Emilio Cruciani, Hlafo~Alfie Mimun, Matteo Quattropani, and Sara Rizzo.
\newblock Phase transition of the k-majority dynamics in biased communication
  models.
\newblock {\em Distributed Comput.}, 36(2):107--135, 2023.
\newblock URL: \url{https://doi.org/10.1007/s00446-023-00444-2}, \href
  {https://doi.org/10.1007/S00446-023-00444-2}
  {\path{doi:10.1007/S00446-023-00444-2}}.

\bibitem{DAmoreCN20}
Francesco D'Amore, Andrea Clementi, and Emanuele Natale.
\newblock Phase transition of a non-linear opinion dynamics with noisy
  interactions - (extended abstract).
\newblock In Andr{\'{e}}a~Werneck Richa and Christian Scheideler, editors, {\em
  Structural Information and Communication Complexity - 27th International
  Colloquium, {SIROCCO} 2020, Paderborn, Germany, June 29 - July 1, 2020,
  Proceedings}, volume 12156 of {\em Lecture Notes in Computer Science}, pages
  255--272. Springer, 2020.
\newblock \href {https://doi.org/10.1007/978-3-030-54921-3\_15}
  {\path{doi:10.1007/978-3-030-54921-3\_15}}.

\bibitem{DAmoreCN22}
Francesco D'Amore, Andrea Clementi, and Emanuele Natale.
\newblock Phase transition of a nonlinear opinion dynamics with noisy
  interactions.
\newblock {\em Swarm Intell.}, 16(4):261--304, 2022.
\newblock URL: \url{https://doi.org/10.1007/s11721-022-00217-w}, \href
  {https://doi.org/10.1007/S11721-022-00217-W}
  {\path{doi:10.1007/S11721-022-00217-W}}.

\bibitem{DAmoreZ22}
Francesco D'Amore and Isabella Ziccardi.
\newblock Phase transition of the 3-majority dynamics with uniform
  communication noise.
\newblock In Merav Parter, editor, {\em Structural Information and
  Communication Complexity - 29th International Colloquium, {SIROCCO} 2022,
  Paderborn, Germany, June 27-29, 2022, Proceedings}, volume 13298 of {\em
  Lecture Notes in Computer Science}, pages 98--115. Springer, 2022.
\newblock \href {https://doi.org/10.1007/978-3-031-09993-9\_6}
  {\path{doi:10.1007/978-3-031-09993-9\_6}}.

\bibitem{dAmoreZ25}
Francesco D'Amore and Isabella Ziccardi.
\newblock Phase transition of the 3-majority opinion dynamics with noisy
  interactions.
\newblock {\em Theor. Comput. Sci.}, 1028:115030, 2025.
\newblock URL: \url{https://doi.org/10.1016/j.tcs.2024.115030}, \href
  {https://doi.org/10.1016/J.TCS.2024.115030}
  {\path{doi:10.1016/J.TCS.2024.115030}}.

\bibitem{DArchivioV24}
Niccol{\`{o}} D'Archivio and Robin Vacus.
\newblock On the limits of information spread by memory-less agents.
\newblock In Dan Alistarh, editor, {\em 38th International Symposium on
  Distributed Computing, {DISC} 2024, October 28 to November 1, 2024, Madrid,
  Spain}, volume 319 of {\em LIPIcs}, pages 18:1--18:21. Schloss Dagstuhl -
  Leibniz-Zentrum f{\"{u}}r Informatik, 2024.
\newblock URL: \url{https://doi.org/10.4230/LIPIcs.DISC.2024.18}, \href
  {https://doi.org/10.4230/LIPICS.DISC.2024.18}
  {\path{doi:10.4230/LIPICS.DISC.2024.18}}.

\bibitem{dubhashi2009}
Devdatt~P. Dubhashi and Alessandro Panconesi.
\newblock {\em Concentration of Measure for the Analysis of Randomized
  Algorithms}.
\newblock Cambridge University Press, 2009.
\newblock URL: \url{http://www.cambridge.org/gb/knowledge/isbn/item2327542/}.

\bibitem{FeinermanHK14}
Ofer Feinerman, Bernhard Haeupler, and Amos Korman.
\newblock Breathe before speaking: efficient information dissemination despite
  noisy, limited and anonymous communication.
\newblock In Magn{\'{u}}s~M. Halld{\'{o}}rsson and Shlomi Dolev, editors, {\em
  {ACM} Symposium on Principles of Distributed Computing, {PODC} '14, Paris,
  France, July 15-18, 2014}, pages 114--123. {ACM}, 2014.
\newblock \href {https://doi.org/10.1145/2611462.2611469}
  {\path{doi:10.1145/2611462.2611469}}.

\bibitem{FeinermanK17}
Ofer Feinerman and Amos Korman.
\newblock The {ANTS} problem.
\newblock {\em Distributed Comput.}, 30(3):149--168, 2017.
\newblock URL: \url{https://doi.org/10.1007/s00446-016-0285-8}, \href
  {https://doi.org/10.1007/S00446-016-0285-8}
  {\path{doi:10.1007/S00446-016-0285-8}}.

\bibitem{FraigniaudKR16}
Pierre Fraigniaud, Amos Korman, and Yoav Rodeh.
\newblock Parallel exhaustive search without coordination.
\newblock In Daniel Wichs and Yishay Mansour, editors, {\em Proceedings of the
  48th Annual {ACM} {SIGACT} Symposium on Theory of Computing, {STOC} 2016,
  Cambridge, MA, USA, June 18-21, 2016}, pages 312--323. {ACM}, 2016.
\newblock \href {https://doi.org/10.1145/2897518.2897541}
  {\path{doi:10.1145/2897518.2897541}}.

\bibitem{FraigniaudN19}
Pierre Fraigniaud and Emanuele Natale.
\newblock Noisy rumor spreading and plurality consensus.
\newblock {\em Distributed Computing}, 32(4):257--276, 2019.
\newblock URL: \url{https://doi.org/10.1007/s00446-018-0335-5}, \href
  {https://doi.org/10.1007/S00446-018-0335-5}
  {\path{doi:10.1007/S00446-018-0335-5}}.

\bibitem{FuggerNR24}
Matthias F{\"{u}}gger, Thomas Nowak, and Joel Rybicki.
\newblock Majority consensus thresholds in competitive lotka-volterra
  populations.
\newblock In Ran Gelles, Dennis Olivetti, and Petr Kuznetsov, editors, {\em
  Proceedings of the 43rd {ACM} Symposium on Principles of Distributed
  Computing, {PODC} 2024, Nantes, France, June 17-21, 2024}, pages 76--86.
  {ACM}, 2024.
\newblock \href {https://doi.org/10.1145/3662158.3662823}
  {\path{doi:10.1145/3662158.3662823}}.

\bibitem{GhaffariL18}
Mohsen Ghaffari and Johannes Lengler.
\newblock Nearly-tight analysis for 2-choice and 3-majority consensus dynamics.
\newblock In Calvin Newport and Idit Keidar, editors, {\em Proceedings of the
  2018 {ACM} Symposium on Principles of Distributed Computing, {PODC} 2018,
  Egham, United Kingdom, July 23-27, 2018}, pages 305--313. {ACM}, 2018.
\newblock URL: \url{https://dl.acm.org/citation.cfm?id=3212738}.

\bibitem{GiakkoupisTZ24}
George Giakkoupis, Volker Turau, and Isabella Ziccardi.
\newblock Self-stabilizing {MIS} computation in the beeping model.
\newblock In Dan Alistarh, editor, {\em 38th International Symposium on
  Distributed Computing, {DISC} 2024, October 28 to November 1, 2024, Madrid,
  Spain}, volume 319 of {\em LIPIcs}, pages 28:1--28:21. Schloss Dagstuhl -
  Leibniz-Zentrum f{\"{u}}r Informatik, 2024.
\newblock URL: \url{https://doi.org/10.4230/LIPIcs.DISC.2024.28}, \href
  {https://doi.org/10.4230/LIPICS.DISC.2024.28}
  {\path{doi:10.4230/LIPICS.DISC.2024.28}}.

\bibitem{GiakkoupisZ23}
George Giakkoupis and Isabella Ziccardi.
\newblock Distributed self-stabilizing {MIS} with few states and weak
  communication.
\newblock In Rotem Oshman, Alexandre Nolin, Magn{\'{u}}s~M. Halld{\'{o}}rsson,
  and Alkida Balliu, editors, {\em Proceedings of the 2023 {ACM} Symposium on
  Principles of Distributed Computing, {PODC} 2023, Orlando, FL, USA, June
  19-23, 2023}, pages 310--320. {ACM}, 2023.
\newblock \href {https://doi.org/10.1145/3583668.3594581}
  {\path{doi:10.1145/3583668.3594581}}.

\bibitem{210018}
Immanuel~Weihnachten (https://mathoverflow.net/users/23351/immanuel
  weihnachten).
\newblock Lower bound on the probability of guessing the mode in a small
  multinomial sample.
\newblock MathOverflow.
\newblock URL:https://mathoverflow.net/q/210018 (version: 2015-06-24).
\newblock URL: \url{https://mathoverflow.net/q/210018}, \href
  {https://arxiv.org/abs/https://mathoverflow.net/q/210018}
  {\path{arXiv:https://mathoverflow.net/q/210018}}.

\bibitem{KormanV23}
Amos Korman and Robin Vacus.
\newblock Distributed alignment processes with samples of group average.
\newblock {\em {IEEE} Trans. Control. Netw. Syst.}, 10(2):960--971, 2023.
\newblock \href {https://doi.org/10.1109/TCNS.2022.3212640}
  {\path{doi:10.1109/TCNS.2022.3212640}}.

\bibitem{LesfariGP22}
Hicham Lesfari, Fr{\'{e}}d{\'{e}}ric Giroire, and St{\'{e}}phane
  P{\'{e}}rennes.
\newblock Biased majority opinion dynamics: Exploiting graph k-domination.
\newblock In Luc~De Raedt, editor, {\em Proceedings of the Thirty-First
  International Joint Conference on Artificial Intelligence, {IJCAI} 2022,
  Vienna, Austria, 23-29 July 2022}, pages 377--383. ijcai.org, 2022.
\newblock URL: \url{https://doi.org/10.24963/ijcai.2022/54}, \href
  {https://doi.org/10.24963/IJCAI.2022/54} {\path{doi:10.24963/IJCAI.2022/54}}.

\bibitem{mitzenmacher2005}
Michael Mitzenmacher and Eli Upfal.
\newblock {\em Probability and Computing: Randomized Algorithms and
  Probabilistic Analysis}.
\newblock Cambridge University Press, 2005.
\newblock \href {https://doi.org/10.1017/CBO9780511813603}
  {\path{doi:10.1017/CBO9780511813603}}.

\bibitem{Shimizu2025}
Nobutaka Shimizu and Takeharu Shiraga.
\newblock 3-majority and 2-choices with many opinions.
\newblock {\em CoRR}, abs/2503.02426, 2025.
\newblock To appear at PODC 2025.
\newblock URL: \url{https://doi.org/10.48550/arXiv.2503.02426}, \href
  {https://arxiv.org/abs/2503.02426} {\path{arXiv:2503.02426}}, \href
  {https://doi.org/10.48550/ARXIV.2503.02426}
  {\path{doi:10.48550/ARXIV.2503.02426}}.

\bibitem{vacus-ziccardi2025}
Robin Vacus and Isabella Ziccardi.
\newblock Minimal leader election under weak communication.
\newblock {\em CoRR}, abs/2502.12697, 2025.
\newblock To appear at PODC 2025.
\newblock URL: \url{https://doi.org/10.48550/arXiv.2502.12697}, \href
  {https://arxiv.org/abs/2502.12697} {\path{arXiv:2502.12697}}, \href
  {https://doi.org/10.48550/ARXIV.2502.12697}
  {\path{doi:10.48550/ARXIV.2502.12697}}.

\end{thebibliography}

\appendix

\section{Missing proofs}\label{sec:missing-proofs}

\subsection{\texorpdfstring{Omitted proofs from \cref{sec:intro:sketch}}{Omitted proofs 0}}\label{sec:omitted:preliminaries}

\begin{proof}[Proof of \cref{lemma:from_lemma_9}]
    We note that
    \begin{align*}
        &\Pr\pa{ X_1 > X_2} - \Pr\pa{ X_2 > X_1}\\
        = \ & \Pr\pa{ X_1 \geq X_2} - \pr{X_1 = X_2} - \pa{\Pr\pa{ X_2 \geq X_1} - \pr{X_1 = X_2}} \\
        = \ & \Pr\pa{ X_1 \geq X_2} - \Pr\pa{ X_2 \geq X_1}
        \\
        = \ & \sum_{k = \lceil h/2 \rceil}^{h} \binom{h}{k} \, p^{k} (1-p)^{h-k} - \sum_{k = \lceil h/2 \rceil}^{h} \binom{h}{k} \, p^{h-k} (1-p)^{k}.
    \end{align*}
    Then the proof continues exactly as in~\cite[Lemma 9]{FraigniaudN19}. Please note that in~\cite[Lemma 9]{FraigniaudN19}, the definition of~$g$ contains a typographical error, i.e., it is defined so that \(g(\delta, h) = \frac{1}{\sqrt{h}} \, \pa{1-\frac{1}{\sqrt{h}}}^{\frac{h-1}{2}}\) if \(\delta \geq \frac{1}{\sqrt{h}}\); 
    the correct definition is the one provided here and is the one actually used in the proof of~\cite[Lemma 9]{FraigniaudN19}.
\end{proof}

\subsection{\texorpdfstring{Omitted proofs from \Cref{subsection:reduce2opinions}}{Omitted proofs 1}}\label{sec:omitted:reduce2opinions}

\begin{proof}[Proof of \Cref{lemma: reduction to 2 opinions}]
The proof of \Cref{lemma: reduction to 2 opinions} requires some minor technical result, which we express through the following claim.

\begin{claim}
    \label{claim: monotonicity_max_binomial_inequality}
    For any integer \(m > 0\), let $Y_1\sim \binomial(m,q)$ and $Y_2=m-Y_1\sim \binomial(m,1 - q)$, with $q > 1/2$. 
    For any $m/2 \leq i' \leq i \leq m$, it holds that
    \begin{align*}
        & \pr{Y_1 > Y_2 \mid M \geq i } - \pr{Y_2 > Y_1 \mid M \geq i } 
        \\
         \geq \ &  \pr{Y_1 > Y_2 \mid M \geq i' } - \pr{Y_2 > Y_1 \mid M \geq i' }.
    \end{align*}
\end{claim}
By \Cref{claim: monotonicity_max_binomial_inequality}, we have
    \begin{align*}
        &\Pr\left( Y_1>Y_2 \mid \max(Y_1,Y_2)>x \right) - \Pr\left( Y_2>Y_1 \mid \max(Y_1,Y_2)>x \right)\\
        = \ & \Pr\left( Y_1>Y_2 \mid \max(Y_1,Y_2)\geq x - 1 \right) 
        \\
        & - \Pr\left( Y_2>Y_1 \mid \max(Y_1,Y_2)\geq x - 1 \right) \\
        \geq \ &  \Pr\left( Y_1>Y_2 \mid \max(Y_1,Y_2) \geq m/2 \right)
        \\
        & - \Pr\left( Y_2>Y_1 \mid \max(Y_1,Y_2)\geq m/2 \right) 
        &\pa{\text{by \Cref{claim: monotonicity_max_binomial_inequality}}}\\
        = \ & \Pr\left( Y_1>Y_2 \right) - \Pr\left( Y_2>Y_1 \right)\\
        \geq \ & C_1 \cdot \min(\sqrt{m}\cdot (2q-1),1),
        &\pa{ \text{by \Cref{lemma:from_lemma_9}} }
    \end{align*}
    for some constant $C_1>0$.
\end{proof}

\begin{proof}[Proof of \Cref{claim: monotonicity_max_binomial_inequality}]
In order to prove \Cref{claim: monotonicity_max_binomial_inequality}, let us first state the following result.
\begin{claim}
    \label{claim: monotonicity_max_binomial}
    For any integer \(m > 1\), let $Y_1\sim \binomial(m,q)$ and $Y_2=m-Y_1\sim \binomial(m,1-q)$, with $q > 1/2$. Furthermore, let $M=\max\cpa{Y_1,Y_2}$. 
    For all integers \(i, i'\) such that $\lceil m/2 \rceil \le i' \le i \le m $, it holds that
    \begin{align*}
        & \pr{Y_1 > Y_2 \mid M=i } - \pr{Y_2 > Y_1 \mid M=i } 
        \\
        \geq & \ \pr{Y_1 > Y_2 \mid M=i' } - \pr{Y_2 > Y_1 \mid M=i' }.
    \end{align*}
\end{claim}
Now we are ready to prove \Cref{claim: monotonicity_max_binomial_inequality}.
    For the sake of brevity, we will denote \[b_j= \pr{Y_1 > Y_2 \mid M = j } - \pr{Y_2 > Y_1 \mid M = j }.\] 
    By \Cref{claim: monotonicity_max_binomial}, we have
    \begin{equation}
        \label{eq:monotonicity_vector}
        b_j \geq b_{j'}, \quad \text{for all $m\geq j \geq j' \geq \lceil m/2 \rceil$.}
    \end{equation}
    We have
    \begin{align*}
        & \pr{Y_1 > Y_2 \mid M \geq i } - \pr{Y_2 > Y_1 \mid  M \geq i } - \pr{Y_1 > Y_2 \mid M \geq i' } 
        \\
         + \ &\pr{Y_2 > Y_1 \mid  M \geq i' }\\
        \underset{(i)}{=} \ & \sum_{j=i}^m \left[\pr{ M = j \mid M \geq i} - \pr{ M = j \mid M \geq i'}\right] b_j - \sum_{j=i'}^{i-1} \pr{ M = j \mid M \geq i'} b_j
        \\
        \underset{(ii)}{\geq}  \ & \sum_{j=i}^m \left[\pr{ M = j \mid M \geq i} - \pr{ M = j \mid M \geq i'}\right] b_j
        \\
        & \ - b_{i-1} \sum_{j=i'}^{i-1} \pr{ M = j \mid M \geq i'} \\
        \underset{(iii)}{=} \ & \sum_{j=i}^m \left[\pr{ M = j \mid M \geq i} - \pr{ M = j \mid M \geq i'}\right] b_j
        \\
         & \ - b_{i-1} \pa{1-\sum_{j=i}^{m} \pr{ M = j \mid M \geq i'}} \\
        = \ &  \sum_{j=i}^m \pr{ M = j \mid M \geq i} b_j - b_{i-1} +  \sum_{j=i}^{m} \pr{ M = j \mid M \geq i'} (b_{i-1}-b_j) \\
        \underset{(iv)}{=} \ & \sum_{j=i}^m \pr{ M = j \mid M \geq i} (b_j - b_{i-1}) +  \sum_{j=i}^{m} \pr{ M = j \mid M \geq i'} (b_{i-1}-b_j) \\
        = \ & \sum_{j=i}^m \left[\pr{ M = j \mid M \geq i} - \pr{ M = j \mid M \geq i'} \right] (b_j - b_{i-1}) \\
        \underset{(v)}{\geq} \ & 0,
    \end{align*}
    where in $(i)$ we used the law of total probabilities, in $(ii)$ we used \Cref{eq:monotonicity_vector}, in $(iii)$ and $(iv)$ we used that $\sum_{j=i}^{m} \pr{ M = j \mid M \geq i}=1$, and in $(v)$ we used \Cref{eq:monotonicity_vector} and that 
    \begin{align*}
        \pr{M = j \st M \ge i} & = \frac{\pr{M = j}}{\pr{M \ge i}} \\
        & \ge \frac{\pr{M = j}}{\pr{M \ge i'}} \\
        & = \pr{M = j \st M \ge i'}. 
    \end{align*}
This concludes the proof of \Cref{claim: monotonicity_max_binomial_inequality}.
\end{proof}

\begin{proof}[Proof of \Cref{claim: monotonicity_max_binomial}]
    Let \(q_1 = q\) and \(q_2 = 1 - q\).
    For all integers $i>m/2$, we have $\cpa{M=i} = \cpa{Y_1 = i} \cup \cpa{Y_2 = i}$, with $\cpa{Y_1 = i} \cap \cpa{Y_2 = i}=\emptyset$. Hence, we obtain
    \begin{align*}
        \Pr(Y_1> Y_2 \mid M = i ) 
        &=
        \frac{\Pr(Y_1 > Y_2, M=i) }{\Pr(M = i) }
        \\
        &=
        \frac{\Pr(Y_1 = i, M=i) }{\Pr(M = i) }
        \\
        &=
        \frac{\binom{m}{i}q_1^i q_2^{m-i}}
        {\binom{m}{i}q_1^i q_2^{m-i} + {\binom{m}{i}q_2^i q_1^{m-i}}}
        \\
        &=
        \frac{q_1^{2i-m}}
        {q_1^{2i-m} + q_2^{2i-m}},
    \end{align*}
    where, in the second inequality, we used that $\cpa{Y_1>Y_2, M=i}= \cpa{Y_1=i, M=i}$. In the symmetrical way, we obtain \[\Pr(Y_2> Y_1 \mid M = i ) 
        =\frac{q_2^{2i-m}}
        {q_1^{2i-m} + q_2^{2i-m}}.\]
    Then, for all integers $m/2 < i' \le i \le m $, we have
    \begin{align}
        \nonumber
        &\Pr(Y_1 > Y_2 \mid M = i )
        -
        \Pr(Y_2 > Y_1 \mid M = i )
        \\
        \nonumber
        = & \
        \frac{q_1^{2i-m} - q_2^{2i-m}}
        {q_1^{2i-m} + q_2^{2i-m}}
        \\
        \nonumber
        = & \
        \frac{(q_1/q_2)^{2i-m} - 1}
        {(q_1/q_2)^{2i-m} + 1}
        \\
        > & \ \label{eq:for_claim_monotonicity_1}
        \frac{(q_1/q_2)^{2i'-m} - 1}
        {(q_1/q_2)^{2i'-m} + 1}
        \\ \label{eq:for_claim_monotonicity_2}
        = & \
        \Pr(Y_1 > Y_2 \mid M = i' )
        -
        \Pr(Y_2 > Y_1 \mid M = i' )
        ,
    \end{align}
    where the inequality holds because the function $\frac{x-1}{x+1}$ is increasing in $x$, and $q_1/q_2>1$ ($q_1>1/2$). 
    
    It remains to prove the inequality in the case $i'=m/2$, for $m$ even. By \Cref{eq:for_claim_monotonicity_1}, we deduce that for $i>m/2$, \[\Pr(Y_1 > Y_2 \mid M = i )
        -
        \Pr(Y_2 > Y_1 \mid M = i )>0,\] because $\frac{x-1}{x+1}>0$ for all $x>1$.
    If $m$ is even, we have that $\cpa{M=m/2} = \cpa{Y_1 = Y_2} $. This implies that $\pr{Y_1 > Y_2 \mid M=m/2 }=\pr{Y_2 > Y_1 \mid M=m/2 }=0$, and, therefore, \[\pr{Y_1 > Y_2 \mid M=m/2 } - \pr{Y_2 > Y_1 \mid M=m/2 } =0.\]
    This implies that for all $i>m/2$, 
    \begin{align*}
        & \Pr(Y_1 > Y_2 \mid M = i )
        -
        \Pr(Y_2 > Y_1 \mid M = i ) \\
        > \ & \Pr(Y_1 > Y_2 \mid M = m/2 )  -
        \Pr(Y_2 > Y_1 \mid M = m/2 ),
    \end{align*}
        which, together with \Cref{eq:for_claim_monotonicity_2}, concludes the proof of \Cref{claim: monotonicity_max_binomial}.
\end{proof}

\begin{proof}[Proof of \Cref{lemma:difference_maj_is_difference_comparison}]
    Our claim is equivalent to proving that
    \begin{align*}
        &\Pr\left( \WW_{1}, \winstrict{1,2} \right) - \Pr\left( \WW_{2}, \winstrict{1,2} \right) 
        \\
        = \ & \Pr\left( X_1 > X_2, \winstrict{1,2} \right) - \Pr\left( X_2 > X_1, \winstrict{1,2} \right).
    \end{align*}
    By definition (\Cref{def:W_1-2_definition}), we have that the event $\cpa{\WW_{1}, \winstrict{1,2}}$ is disjoint union of the events $\cpa{X_1 >X_2, \winstrict{1,2}}$ and the event \[\cpa{X_1 = X_2, \winstrict{1,2}, \text{"opinion 1 gets chosen among the ties"}}.\] It holds similarly for the event $\cpa{\WW_{2}, \winstrict{1,2}}$.
    Therefore, we can write
    \begin{align*}
                    &\Pr\left( \WW_{1}, \winstrict{1,2} \right) - \Pr\left( \WW_{2}, \winstrict{1,2} \right)\\
                    = \ & \Pr\left( X_1>X_2, \winstrict{1,2} \right) + \frac{1}{2}\Pr\left( X_1=X_2, \winstrict{1,2}  \right)
                    \\
                    & - \Pr\left( X_2>X_1, \winstrict{1,2} \right) - \frac{1}{2}\Pr\left( X_1=X_2 , \winstrict{1,2} \right)\\
                    = \ & \Pr\left( X_1 > X_2 \mid \winstrict{1,2} \right) - \Pr\left( X_2 > X_1 \mid \winstrict{1,2} \right). \tag*{\qedhere} 
    \end{align*}
\end{proof}

\begin{proof}[Proof of \Cref{lemma:sum_12_is_independent_maximum}]

    We found it more convenient to swap the conditioning event with the event whose probability we are computing, and then apply Bayes' theorem. Hence, we first state the following claim. 
    \begin{claim} 
    \label{claim:for lemma:sum_12_is_independent_maximum}
                    For any $0\leq x<h$,
                    \[
                        \Pr(\winstrict{1,2}\mid X_1 + X_2 \geq x)\geq \Pr(\winstrict{1,2}\mid X_1 + X_2 \geq x-1).
                    \]
\end{claim}
                Now, by Bayes theorem and by \Cref{claim:for lemma:sum_12_is_independent_maximum}, we have
                \begin{align*}
                    & \ \Pr( X_1 + X_2 \geq h\cdot (p_1+p_2)/2 \mid \winstrict{1,2})
                    \\
                    = & \
                    \frac{\Pr(X_1 + X_2\geq h\cdot (p_1+p_2)/2)\cdot \Pr(\winstrict{1,2} \mid X_1 +X_2 \geq h\cdot (p_1+p_2)/2 )}{\Pr(\winstrict{1,2})} \\
                    \geq & \ 
                    \frac{\Pr(X_1 + X_2\geq h\cdot (p_1+p_2)/2)\cdot \Pr(\winstrict{1,2} \mid X_1 +X_2 \geq 0 )}{\Pr(\winstrict{1,2})} \\
                    = & \ 
                    \Pr(X_1 + X_2\geq h\cdot (p_1+p_2)/2),
                \end{align*}
                concluding the proof of \Cref{lemma:sum_12_is_independent_maximum}.
\end{proof}

\begin{proof}[Proof of \Cref{claim:for lemma:sum_12_is_independent_maximum}]

We want to use a coupling argument. The idea is that \[\Pr(\winstrict{1,2}\mid X_1 + X_2 \geq x)\] is the probability that we have $\winstrict{1,2}$ after $x$ extractions among $\{1,2\}$ only, and the rest of them are among all opinions. More specifically we consider two families of r.v. representing $h$ extractions among balls with $k$ opinions. For both families, the first $x-1$ extractions are identical and only among the first two opinions, and from the $x+1$ to the end are identical and among all opinions. They instead differ on the $x$-th extraction. For first family the $x$-th extraction is simply among the first two opinions. For the second family, we need to flip first an unbalanced coin: in case of head, the extraction is identical to the first family,  in case of tail, the extraction is on all opinions except the first twos. The coin is unbalanced in such a way the marginal distribution of the $x$-th extraction of the second family is like the one of an extraction among all opinions.
Let $\{ Y_j \}_{j \in \{1,\dots, h\}}$ family of independent r.v. with the following distributions:
                     For $0 \leq j \leq x$:
                    \begin{align*}
                        \Pr( Y_j = 1 ) &= \frac{p_1}{p_1+p_2};\\
                        \Pr( Y_j = 2 ) &= \frac{p_2}{p_1+p_2}.
                    \end{align*}
                    For $x+1 \leq j \leq h$ and $i\in\{1,\dots, k\}$:
                    \begin{align*}
                        \Pr( Y_j = i ) &= p_i.
                    \end{align*}
                    Now let $\{ \overline Y_j \}_{j \in \{1,\dots, h\}}$ be the r.v. family s.t. $\overline Y_j = Y_j$ for all $j\neq x$, and $\overline Y_x$ with the following distribution. Let $C$ a random coin, independent of  $\{ Y_j \}_{j \in \{1,\dots, h\}}$, with distribution $\Pr(C=1)=p_1+p_2$ and $\Pr(C=0)=1-(p_1+p_2)$. $\overline Y_x$ is distributed in such a way we obtain                    \begin{align*}
                        \Pr( C=1, \overline Y_x = Y_x  ) &= 1
                    \end{align*}
                    For $i\in\{3,\dots, k\}$,
                    \[
                        \Pr( C=0, \overline Y_x = i ) = \frac{p_i}{\sum_{j \geq 3} p_i}.
                    \]
                    Overall, we have, for $i\in\{1,\dots, k\}$
                    \begin{align*}
                        \Pr( \overline Y_x = i ) &= p_i.
                    \end{align*}
                    Now define $\{ Z_i \}_{i \in \{1,\dots, k\}}$ and $\{ \overline Z_i \}_{i \in \{1,\dots, k\}}$ as follows:
                    for all $i \in \{1,\dots, k\}$
                    \begin{align*}
                        Z_i  &=  \abs{\{ j \in \{1,\dots, h\} : Y_j = i \}} \\
                        \overline Z_i  &=  \abs{\{ j \in \{1,\dots, h\} : \overline Y_j = i \}}.
                    \end{align*}
                    In this way, we have that $(X_i \mid X_1 + X_2 \geq x ) \sim Z_i$ and $(X_i \mid X_1 + X_2 \geq x-1 ) \sim \overline Z_i$ and then
                    \begin{align*}
                        & \Pr(\winstrict{1,2}\mid X_1 + X_2 \geq x) = \Pr ( \cap_{i\geq 3} \{ \max( Z_1, Z_2) > Z_i \} ) \text{ and} \\
                        & \Pr(\winstrict{1,2}\mid X_1 + X_2 \geq x-1) = \Pr ( \cap_{i\geq 3} \{ \max( \overline Z_1, \overline Z_2) > \overline Z_i \} ).
                    \end{align*}
                    For our construction, we have for all $i \in \{1,\dots, k\}$
                    \begin{align*}
                        \cpa{ C = 1 } \implies \cpa{ Z_i = \overline Z_i } 
                    \end{align*}
                    and for $i=1,2$
                    \begin{align*}
                        \cpa{ C = 0 } \implies \cpa{ Z_i \geq \overline Z_i } 
                    \end{align*}
                    while for $i\geq 3$
                    \begin{align*}
                        \cpa{ C = 0 } \implies \cpa{ Z_i \leq \overline Z_i } ,
                    \end{align*}
                    These facts imply, regardless the realization of $C$, \[\{\cap_{i\geq 3} \{ \max ( \overline Z_1, \overline Z_2) > \overline Z_i \}\} \subset \{\cap_{i\geq 3} \{ \max ( Z_1, Z_2) > Z_i \}\}. \]
                    We conclude
                    \begin{align*}
                        \Pr(\winstrict{1,2}\mid X_1 + X_2 \geq x) &= \Pr ( \cap_{i\geq 3} \{ \max( Z_1, Z_2) > Z_i \} )    \\
                        &\geq
                        \Pr ( \cap_{i\geq 3} \{ \max( \overline Z_1, \overline Z_2) > \overline Z_i \} ) \\
                        &=
                        \Pr(\winstrict{1,2}\mid X_1 + X_2 \geq x-1).
                    \end{align*}
                This concludes the proof of \Cref{claim:for lemma:sum_12_is_independent_maximum}.
                \end{proof}

\subsection{\texorpdfstring{Omitted proofs from \Cref{sec:analysis:ties-estimation}}{Omitted proofs 2}}\label{sec:omitted:ties-estimation}

\begin{proof}[Proof of \Cref{claim:weak_opinion_lose_whp}]
    Define the event \(\CC_j = \cpa{\abs{X_j - hp_j } < \sqrt{3 C_3 \, hp_j \log n}}\) and \(\CC = \cap_{j =1}^k \CC_j\).
    The event $\mathcal{C}_j$ can be rewritten as 
    \begin{equation}
        \label{eq:interval_rare_opinions}
        \mathcal{C}_j = \cpa{ X_j \in \pa{ hp_j - \sqrt{3 C_3 \, hp_j \log n} , hp_j + \sqrt{3 C_3 \, hp_j \log n} } }.
    \end{equation} For $i\in \mathcal{R}$, we have 
    \begin{align*}
        & \ h p_1 - \sqrt{3 \, C_3 \, hp_1 \log n} - \pa{h p_i + \sqrt{3 \, C_3 \, hp_i \log n}}
        \\
        \geq & \ \pa{1-C_2} h p_1 - \sqrt{3 \, C_3 \, hp_1 \log n}
        & \pa{ i\in \mathcal{R} \implies h p_i \leq h C_2 \cdot p_1 }\\
        > & \ 0,
    \end{align*}
    by taking $hp_1 > C_4 \log n$, with $C_4 = \pa{3\, C_3\frac{1}{1-C_2}}^2$. This fact, together with \Cref{eq:interval_rare_opinions}, implies that $\CC \subset \cpa{\cap_{i\in\mathcal{R}_{C_2}} \cpa{X_1 > X_i}}$, and therefore $\pr{\cap_{i\in\mathcal{R}_{C_2}} \cpa{X_1 > X_i}} \geq \pr{\CC}$. 
    By the multiplicative Chernoff bound we have
    \begin{align*}
        \pr{C_j}
        &= 
        \pr{\abs{X_j - hp_j } < \sqrt{3 C_3 \, hp_j \log n}}
        \\
        &= 
        1 - \pr{\abs{X_j - hp_j } \geq \sqrt{3 C_3 \, hp_j \log n}}
        \\
        &\geq 
        1 - 2 \exp\pa{ -\frac{3 C_3 \log n}{3} }
        \\
        &=
        1-\frac{2}{n^{C_3}},
    \end{align*}
    and by the union bound and by the fact the number of opinion $k\leq n$, we obtain \[ \pr{\cap_{i\in\mathcal{R}} \cpa{X_1 > X_i}} \geq \pr{\CC} \ge 1 - \frac{2 k}{n^{C_3}}\geq 1 - \frac{1}{n^{C_3-2}}, \]
    concluding the proof of \Cref{claim:weak_opinion_lose_whp}.
\end{proof}

\begin{proof}[Proof of \Cref{lemma:bound_W_1-2}]
    Recall $\mathcal{S}:= \cpa{i\in [k] : p_i > \frac{1}{2}p_1 }$ be the set of the \textit{strong opinions}. By definition of $\mathcal{S}$ we have
    \begin{align*}
        1 \geq \sum_{j\in \mathcal{S}} p_j > \frac{\abs{\mathcal{S}} \, p_1 }{2},
    \end{align*}
    and therefore we have
    \begin{equation}
        \label{eq:number_strong_opinions}
        \abs{\mathcal{S}} \leq \frac{2}{p_1}.
    \end{equation}
    Recall that $\mathcal{R}_{1/2}:= \cpa{i\in [k] : p_i\leq \frac{1}{2}p_1 }$ is the set of the $1/2$-\textit{rare opinions}. If $i\in \mathcal{R}_{1/2}$, we have that \[\pr{\WW_i \mid \cap_{\ell\in\mathcal{R}_{1/2}} \cpa{X_1 > X_\ell}} = 0.\] By this fact, and by noticing that an opinion is either strong or $1/2$-rare,
    \begin{equation}
        \label{eq:only_strong_opinions_win}
        1 = \pr{\bigcup_{j\in \mathcal{S}} \WW_j \mid \cap_{\ell\in\mathcal{R}_{1/2}} \cpa{X_1 > X_\ell}} \leq \sum_{j\in \mathcal{S}} \pr{\WW_j \mid \cap_{\ell\in\mathcal{R}_{1/2}} \cpa{X_1 > X_\ell}}.
    \end{equation}
    Since $X_j\sim \binomial(p_j,h)$ and $p_1\geq p_j$, for all $j\in [k]$, we have that $X_1$ stochastically dominates $X_j$ for all $j\geq 2$. This fact remains true also conditioning on the event $\cpa{\cap_{\ell\in\mathcal{R}_{1/2}} \cpa{X_1 > X_\ell}}$. Therefore, we have \[\pr{\WW_1 \mid \cap_{\ell\in\mathcal{R}_{1/2}} \cpa{X_1 > X_\ell} } \geq \pr{\WW_j \mid \cap_{\ell\in\mathcal{R}_{1/2}} \cpa{X_1 > X_\ell} },\] for all $j \in \mathcal{S}, j\geq 2 $. This fact, together with \Cref{eq:only_strong_opinions_win}, implies
    \[
        \abs{\mathcal{S}}\cdot \pr{\WW_1 \mid \cap_{\ell\in\mathcal{R}_{1/2}} \cpa{X_1 > X_\ell} } \geq \sum_{j\in \mathcal{S}} \pr{\WW_j \mid \cap_{\ell\in\mathcal{R}_{1/2}} \cpa{X_1 > X_\ell}} \geq 1,
    \]
    and, by using \Cref{eq:number_strong_opinions}, we obtain
    \[
        \pr{\WW_1 \mid \cap_{\ell\in\mathcal{R}_{1/2}} \cpa{X_1 > X_\ell} } \geq \frac{1}{\abs{\mathcal{S}}} \geq \frac{p_1}{2}.
    \]
    We can conclude the proof of \Cref{lemma:bound_W_1-2} noticing that, by \Cref{claim:weak_opinion_lose_whp} with $C_3=3$ and $C_2=1/2$ (which implies \(C_4 = (18)^2\), we have
    \[
        \Pr(\WW_1) \geq \pr{\cap_{\ell\in\mathcal{R}_{1/2}} \cpa{X_1 > X_\ell}} \cdot \pr{\WW_1 \mid \cap_{\ell\in\mathcal{R}_{1/2}} \cpa{X_1 > X_\ell} } \geq \pa{1-\frac{1}{n}} \frac{p_1}{2} \geq \frac{p_1}{3},
    \]
    for $n$ large enough.
\end{proof}

\begin{proof}[Proof of \cref{lemma:relationship_ties_without_ties}]
Define the event \(\CC_j = \cpa{\abs{X_j - hp_j } < \sqrt{9hp_j \log n}}\) and \(\CC = \cap_{j \ge 1} \CC_j\).
By the multiplicative Chernoff bound  (\Cref{lemma:app:multiplicative-chernoff} in \Cref{sec:tools}), we obtain
\begin{align*}
        \pr{\CC_j}
        &= 
        \pr{\abs{X_j - hp_j } < \sqrt{9 \, hp_j \log n}}
        \\
        &= 
        1 - \pr{\abs{X_j - hp_j } \geq \sqrt{9 \, hp_j \log n}}
        \\
        &\geq 
        1 - 2 \exp\pa{ -\frac{9 \log n}{3} }
        \\
        &=
        1-\frac{2}{n^{3}}.
    \end{align*}
Moreover, by the union bound and by the fact $k\leq n$, we have, for $n$ large enough, 
\begin{equation}
    \label{eq:concentration_strong_opinions}
    \pr{\CC}
\ge 1 - \frac{2k}{n^{3}}\geq 1 - \frac{1}{n}. 
\end{equation}
Recall $\mathcal{S}:= \cpa{i\in [k] : p_i > \frac{1}{2}p_1 }$ be the set of the \textit{strong opinions}.
We first note that 
\begin{align}
    \nonumber
    \CC
    &\subseteq 
    \bigcap_{j\in\mathcal{S}}\cpa{X_j>hp_j-\sqrt{9hp_j \log n} } 
    \\ \nonumber
    &\subseteq 
    \bigcap_{j\in\mathcal{S}}\cpa{X_j>\frac{1}{2} hp_1-\sqrt{9hp_1 \log n} }
    & \pa{ p_i>\frac{1}{2}p_1, p_i\leq p_1, \text{ for all } i\in \mathcal{S}}
    \\ \label{eq:strong_opinion_score_at_least_1}
    &\subseteq 
    \cap_{j\in\mathcal{S}}\cpa{X_j>0}.
    & \pa{ hp_1 > 36 \log n }
\end{align}

Let us denote a particular realization of a multinomial in the following way.
\begin{equation*}
    \mathcal{X}(x_1,\dots,x_k)=\cpa{X_1=x_1,\dots, X_k=x_k}.
\end{equation*}

We define the event $\mathcal{X}(x_1,\dots,x_k)$ a \textbf{1-tie} if $x_1 \geq x_i$ for all $i\geq 1$ and there exists a $j \neq 1$ s.t. $x_j=x_1$. In other words, opinion 1 is the most sampled, together with at least another opinion. We denote the set of 1-tie as $\mathcal{T}_1$.

We map a 1-tie to an event in which the opinion 1 gets one extra sample which is stolen from the least scoring strong opinion. Formally, if $\mathcal{X}(x_1,\dots,x_k)\in\mathcal{T}_1$ and \(j = \max \cpa{i : X_i = \min_{r \in \mathcal{S}} \cpa{X_r}}\), we define:
\[
    f:\,\mathcal{X}(x_1,\dots,x_k) \mapsto \mathcal{X}(x_1+1,\dots,x_{j-1},x_j - 1, x_{j+1}, \dots ,x_k).
\]
The map is well-defined if and only if $x_j>0$. If $\mathcal{X}(x_1,\dots,x_k) \in \CC$, by \Cref{eq:strong_opinion_score_at_least_1}, we have $\mathcal{X}(x_1,\dots,x_k) \in \cap_{j\in\mathcal{S}}\cpa{X_j>0}$, and, since $j\in \mathcal{S}$ by definition, we obtain that $x_j>0$. Therefore, $f$ is well-defined on $\CC$.

Moreover, the map is injective.
Using the probability mass function of the multinomial distribution, we compute the following:
\[
   \frac{\pr{f\pa{\mathcal{X}(x_1,\dots,x_k)}}}{ \pr{ \mathcal{X}(x_1,\dots,x_k) }} = \frac{ x_j}{x_1+1} \cdot  \frac{p_1}{p_j} \geq \frac{ x_j}{x_1+1},
\]
where we used that $p_1 \geq p_j$.
Therefore we obtain
\begin{align}
    \nonumber
    & \sum_{\substack{\mathcal{X} \in \mathcal{T}_1 , \XX \in \CC }}\pr{f\pa{\mathcal{X}(x_1,\dots,x_k)}} 
    \\ \nonumber
    \geq \ & 
    \sum_{\substack{\mathcal{X} \in \mathcal{T}_1 , \XX \in \CC }}\frac{x_j}{x_1+1} \pr{\pa{\mathcal{X}(x_1,\dots,x_k)}}
    \\ \nonumber
    \geq \ & 
    \frac{ h p_j - \sqrt{9 h p_j \log n} }{1 + h p_1 + \sqrt{9 h p_1 \log n}} \cdot \sum_{\substack{\mathcal{X} \in \mathcal{T}_1 , \XX \in \CC}} \pr{ \mathcal{X}(x_1,\dots,x_k) }
    &\pa{ \text{by definition of }\CC }
    \\  \nonumber
    \geq \ & 
    \frac{ \frac{1}{2} h p_1 - \sqrt{9 h p_1 \log n} }{1 + h p_1 + \sqrt{9 h p_1 \log n}} \cdot \sum_{\substack{\mathcal{X} \in \mathcal{T}_1 , \XX \in \CC}} \pr{ \mathcal{X}(x_1,\dots,x_k) } 
    & \pa{ p_i>\frac{1}{2}p_1, p_i\leq p_1, \text{ for } i\in \mathcal{S}}
    \\  \nonumber
    = \ & 
    \frac{ \frac{1}{2} \sqrt{\frac{hp_1}{9 \log n}}  -1 }{ \sqrt{\frac{hp_1}{9 \log n}}  + 2} \cdot \sum_{\substack{\mathcal{X} \in \mathcal{T}_1 , \XX \in \CC}} \pr{ \mathcal{X}(x_1,\dots,x_k) } 
    \\ \label{eq:relationship_probabilities_1_ties}
    \ge \ & 
    \frac{1}{4}\cdot \sum_{\substack{\mathcal{X} \in \mathcal{T}_1 , \XX \in \CC}} \pr{ \mathcal{X}(x_1,\dots,x_k) },
\end{align}
where the last inequality holds because $ x\mapsto\frac{\frac{1}{2}x-1 }{x+2} \text{ is increasing and } hp_1 \geq 18^2 \log n  $.
Recall \Cref{def:W_1-2_definition}. Since $\mathcal{T}_1 = \WW_{1,\text{ties}}\setminus\WW_{1,\text{strict}}$ and $f\pa{\mathcal{T}_1 \cap \CC} \subseteq \WW_{1}$, we can write
\begin{align*}
    & \ \pr{\WW_{1,\text{ties}}, \CC} 
    \\
    = & \ \pr{\WW_{1,\text{strict}}, \CC} + \pr{\WW_{1,\text{ties}}\setminus\WW_{1,\text{strict}}, \CC}
    \\
    \geq & \ 
    \sum_{\substack{\mathcal{X} \in \mathcal{T}_1, \XX \in \CC}}\pr{f\pa{\mathcal{X}(x_1,\dots,x_k)}} + \sum_{\substack{\mathcal{X} \in \mathcal{T}_1, \XX \in \CC}} \pr{ \mathcal{X}(x_1,\dots,x_k) }
    \\
    \geq & \ 
    \left( 1 + \frac{1}{4} \right) \sum_{\substack{\mathcal{X} \in \mathcal{T}_1,\XX \in \CC}} \pr{ \mathcal{X}(x_1,\dots,x_k) }
    & \pa{\text{by \Cref{eq:relationship_probabilities_1_ties}}}
    \\
    = & \  \left( 1 + \frac{1}{4} \right)\pr{\WW_{1,\text{ties}}\setminus\WW_{1,\text{strict}}, \CC}.
\end{align*}
Hence, 
\[
    \pr{\WW_{1,\text{ties}}\setminus\WW_{1,\text{strict}}, \CC} \leq \frac{4}{5} \cdot \pr{\WW_{1,\text{ties}}, \CC},
\]
and
\begin{align*}
    & \ \pr{\WW_{1,\text{strict}}}
    \ge \pr{\WW_{1,\text{strict}}, \CC}
    =  \pr{\WW_{1,\text{ties}}, \CC} - \pr{\WW_{1,\text{ties}}\setminus\WW_{1,\text{strict}}, \CC}
    \\
    \ge & \ \frac{1}{5} \cdot \pr{\WW_{1,\text{ties}},\CC}
    \\
    = & \ \frac{1}{5} \pa{ \winties{1} + \pr{\CC} - \pr{\WW_{1,\text{ties}} \cup \CC}}
    \qquad \pa{\text{by the inclusion–exclusion princ.}}
    \\
    \ge & \ \frac{1}{5} \pa{ \winties{1} + \pr{\CC} -1 }
   \qquad \qquad \qquad \qquad \, \pa{ \pr{\WW_{1,\text{ties}} \cup \CC} \leq 1 }
    \\
    \ge & \ \frac{1}{5} \pa{ \winties{1} - \frac{1}{n} }
    \qquad \qquad \qquad \qquad \qquad \quad \pa{ \text{by \Cref{eq:concentration_strong_opinions}} }
    \\
    \ge & \ \frac{1}{6} \winties{1}
    \ge \frac{1}{18} p_1
\end{align*}
for $n$ large enough. This concludes the proof of \Cref{lemma:relationship_ties_without_ties}.
\end{proof}

\begin{proof}[Proof of \Cref{lemma:bound_M_1-2}]
By \Cref{lemma:relationship_ties_without_ties,lemma:bound_W_1-2}, we obtain
    \[
        \pr{\winstrict{1,2}} \geq \pr{\WW_{1,\text{strict}}} \geq \frac{1}{6} \Pr(\winties{1}) \geq \frac{1}{18}\, p_1 = \frac{1}{36} \pa{p_1+p_1} \geq \frac{1}{36}(p_1 + p_2).
    \]
This concludes the proof of \Cref{lemma:bound_M_1-2}.
\end{proof}

\subsection{\texorpdfstring{Omitted proofs from \Cref{sec:analysis:alltogether}}{Omitted proofs 3}}\label{sec:omitted:alltogether}

\begin{proof}[Proof of \Cref{proposition:bias_hp>1}]
    Since $p_2\geq p_j$ for all $j\geq 2$, we have that $\pr{\WW_{j}} \leq \pr{\WW_{2}}$, and, therefore,
    \[
         \Pr\left( \WW_{1} \right) - \Pr\left( \WW_{2} \right) \implies \Pr\left( \WW_{1} \right) - \Pr\left( \WW_{j} \right),
    \]
    and,
    \[
         \Pr\left( \WW_{1} \right) \geq \frac{1}{1-C_6} \Pr\left( \WW_{2} \right) \implies \Pr\left( \WW_{1} \right) \geq  \frac{1}{1-C_6} \Pr\left( \WW_{j} \right) .
    \]
    Hence, we will focus on the proof statement for $j=2$.

    Consider the case $\delta(2) \leq \sqrt{\frac{2(p_1+p_2)}{h}}$, that, in particular, implies that \[\min\pa{ \pa{\frac{\delta(2) \sqrt{h}}{\sqrt{2(p_1+p_2)}}},1} =  \pa{\frac{\delta(2) \sqrt{h}}{\sqrt{2(p_1+p_2)}}}.\] We have
    \begin{align*}
        &\Pr\left( \WW_{1} \right) - \Pr\left( \WW_{2} \right)
        \\
        \geq & \  \Pr\left( \winstrict{1,2} \right)\cdot \left[  \Pr\left( \WW_{1} \mid \winstrict{1,2} \right) - \Pr\left( \WW_{2} \mid \winstrict{1,2} \right) \right] \\
        \geq & \  \Pr\left( \winstrict{1} \right)\cdot \left[  \Pr\left( \WW_{1} \mid \winstrict{1,2} \right) - \Pr\left( \WW_{2} \mid \winstrict{1,2} \right) \right] 
        &\pa{\winstrict{1} \subseteq \winstrict{1,2} }
        \\
        \geq & \  \frac{1}{6} \Pr\left( \winties{1} \right)\cdot \left[  \Pr\left( \WW_{1} \mid \winstrict{1,2} \right) - \Pr\left( \WW_{2} \mid \winstrict{1,2} \right) \right] 
        &\pa{\text{by \Cref{lemma:relationship_ties_without_ties}}}
        \\
        \geq & \  \frac{1}{6} \Pr\left( \WW_{1} \right)\cdot \left[  \Pr\left( \WW_{1} \mid \winstrict{1,2} \right) - \Pr\left( \WW_{2} \mid \winstrict{1,2} \right) \right] 
        &\pa{\WW_1 \subseteq \winties{1}}
        \\
        \geq & \ \frac{C}{6} \, \pr{\WW_1} \delta(2)\sqrt{\frac{h}{p_1+p_2}}
        &\pa{\text{by \Cref{lemma:lemma_ema_with_conditioning}}}
    \end{align*}
    Setting $C_5 = \frac{C}{6}$, we conclude the proof of the first statement of \Cref{proposition:bias_hp>1}.
    Now consider the case $\delta \geq \sqrt{\frac{2(p_1+p_2)}{h}}$, that implies that $ \pa{\frac{\delta \sqrt{h}}{\sqrt{2(p_1+p_2)}}}=1$.
    \begin{align*}
        & \ \Pr\left( \WW_{1} \right) - \Pr\left( \WW_{2} \right)
        \\
        \geq & \ \Pr\left( \winstrict{1,2} \right)\cdot \left[  \Pr\left( \WW_{1} \mid \winstrict{1,2} \right) - \Pr\left( \WW_{2} \mid \winstrict{1,2} \right) \right] \\
        \geq & \ \pr{ \WW_{1,2,\text{strict}} } \cdot C
        &\text{(by \Cref{lemma:lemma_ema_with_conditioning})}
        \\
        \geq & \ \pr{ \WW_{1, \text{strict}} } \cdot C
        &\pa{\winstrict{1} \subseteq \winstrict{1}}
        \\
        \geq & \ \frac{C}{6} \pr{\winties{1}}
        &\text{(by \Cref{lemma:relationship_ties_without_ties})}
        \\
        \geq & \ \frac{C}{6} \pr{\WW_{1}}
        &\pa{\WW_{1}\subseteq\winties{1}}
        \\
        \geq & \ C_6 \pr{\WW_{1}},
    \end{align*}
    if we set $0<C_6=\min\pa{\frac{1}{2},\frac{C}{6}}<1$.  This concludes the proof of \Cref{proposition:bias_hp>1}.
\end{proof}

\begin{proof}[Proof of \Cref{thm:hp_1 large}]

    By using the union bound and by \Cref{lemma:disappear_one_step,lemma:disappear_two_steps,lemma:bias_exponential_growth}, we have that either $\delta(j)' \geq e \delta(j)$, or $\frac{p'_j}{p'_1} < \frac{p_j}{p_1}$, or opinion $j$ disappear and $p_j=0$, for all $j\neq 1$ with probability $1- \frac{1}{n^4}$. This means that,conditioning on this event, we can apply, the following claim (whose proof is deferred to \Cref{sec:omitted:alltogether}) to obtain $p'_1\geq p_1$, with probability $1- \frac{1}{n^4}$.
    \begin{claim}               \label{claim:all_bias_increase_p1_increase}
        Let a partition of $[k]=\cpa{1}\cup I\cup J\cup K$ s.t.
        for $j\in I$, $\delta(j)'>\delta(j)$, for $j\in J$, $\frac{p'_j}{p'_1}<\frac{p_j}{p_1}$, and for $j \in K$, $p'_j=0$. Then $p'_1>p_1$. 
    \end{claim}
    This fact, together with \Cref{lemma:new_bias_hypothesis} and the union bound on all opinions, implies that with probability $1-\frac{2}{n^3}$, after one round of \hmaj, the new configuration still satisfies the initial hypothesis of the theorem, and we can iterate our calculations. By the union bound on several rounds, we can iterate this calculations up to $n$ number of rounds to obtain all the events hold w.h.p. 
    
    By \Cref{lemma:bias_exponential_growth}, at every round, with probability $1-1/n^4$, $\delta(j)'\geq e \, \delta(j)$ for $n$ large enough, until we obtain a bias $\delta(j) \geq \sqrt{\frac{2(p_1+p_2)}{h}}$. Then, even starting with the minimal initial bias of $\delta(j)=\sqrt{\log n}/n$, we reach a configuration with $\delta(j)\geq \sqrt{\frac{2(p_1+p_2)}{h}}$ after \[
        \log\pa{ n \sqrt{\frac{2(p_1+p_2)}{h\log n}} }\leq \log(\frac{2 \sqrt{C_4} n}{\log n}) \leq \log(n) 
    \]
    number of rounds, with probability $1-1/n^{3}$, where we used that $h\geq C_4 \frac{\log n}{p_1}$, that $p_2\leq p_1 \leq 1$ and we took $n$ large enough. Once we have $\sqrt{\frac{2(p_1+p_2)}{h}} \leq \delta(j) \leq \pa{1 - \frac{1}{1+C_6}}p_1$, by \Cref{eq:bias_increase_2} we have that at the next round with probability $1-\frac{2}{n^4}$ we are in the last case $\delta(j) \geq \pa{1-\frac{1}{1+C_6}}p_1$, and, as computed before, in the following round opinion $j$ will disappear with probability $1 - \frac{1}{n^4}$.
    
    Overall we have that all opinions $j\neq 1$ disappear with probability $1-\frac{1}{n^3}$ in at most $\log(n) +2$ number of rounds, and by applying the union bound, we obtain that w.h.p.\ after $\log n$ rounds the only opinion left is opinion 1 with probability $1- \frac{1}{n^2}$, concluding the proof of \Cref{thm:hp_1 large}.
\end{proof}

\begin{proof}[Proof of \Cref{lemma:disappear_one_step}]
    Recalling $\mathcal{R}_{x}:= \cpa{i\in [k] : p_i\leq x p_1 }$, we have that $j\in\mathcal{R}_{x}$, and therefore, by \Cref{claim:weak_opinion_lose_whp}, with probability $1-\frac{1}{n^4}$ opinion $j$ will score less than opinion 1, and then will disappear in the next round and therefore $\delta(j)'=p'_1$. 
\end{proof}

\begin{proof}[Proof of \Cref{lemma:disappear_two_steps}]
    We have that
    \begin{align}
        \nonumber
        \pr{\WW_1}
        &\geq \frac{1}{3
        }p_1
        &\pa{\text{by \Cref{lemma:bound_W_1-2}}}
        \\
        \label{eq:p(W_1)>p_1}
        &\geq \frac{C_7 \log(n)}{3 n}.
        &\pa{p_1 \geq \frac{C_7 \log n}{n}}
    \end{align}
    Denote with $Y^{(i)}_\ell$ the indicator function of the event "the agent $i$ adopts opinion $\ell$". In this way, we have that $\sum_{i\in[n]} Y^{(i)}_\ell$ counts the number of agents supporting opinion $\ell$ during the next round, and $\bbE\pa{ \sum_{i\in[n]} Y^{(i)}_\ell } = n \pr{ \WW_\ell }$. By the multiplicative Chernoff bound (\Cref{lemma:app:multiplicative-chernoff} in \Cref{sec:tools}), we obtain
    \begin{align}
        \nonumber
        &\pr{ \abs{\sum_{i\in[n]} Y^{(i)}_1 - n \pr{\WW_1}} \leq \frac{C_6^2}{2} n \pr{ \WW_1 } }
        \\
        \nonumber
        \leq \ & 2 \exp\pa{ -\frac{\pa{\frac{C_6^2}{2}}^2}{2} n \pr{ \WW_{1} } }
        \\        \label{eq:concentration_p1_next_round}
        \leq \ & \frac{1}{n^4}.
        &\pa{ \text{by \Cref{eq:p(W_1)>p_1}}}
    \end{align}
    By \Cref{proposition:bias_hp>1}, we know that $\pr{ \WW_{j} } \leq (1-C_6) \pr{ \WW_{1} }$ for some constant $0<C_6<1$. 
    Let $\cpa{Z_i}_{i\in[n]}$ be i.i.d.\ random variables with Bernoulli distribution of parameter $(1-C_6) \pr{ \WW_{1} }$.
    We obtain that $\sum_{i\in[n]} Y^{(i)}_j$ is 				stochastically dominated by $\sum_{i\in[n]} Z^{(i)}$. Hence,
    \begin{align*}
        &\pr{ \sum_{i\in[n]} Y^{(i)}_j \geq \pa{1+\frac{C_6^2}{2(1-C_6^2)}} n (1-C_6) \pr{ \WW_1 } }
        \\
		\leq \ & \pr{ \sum_{i\in[n]} Z^{(i)} \geq \pa{1+\frac{C_6^2}{2(1-C_6^2)}} n (1-C_6) \pr{ \WW_1 } }
        \\
        \leq \ & \exp\pa{ -\frac{\pa{\frac{C_6^2}{2(1-C_6^2)}}^2 n (1-C_6) \pr{ \WW_1 }}{3} }
        \\
        \leq \ & \frac{1}{n^4}.
        &\pa{ \text{by \Cref{eq:p(W_1)>p_1}}}
    \end{align*}
    Since $1-\frac{C_6^2}{2}\geq \pa{1+\frac{C_6^2}{2(1-C_6^2)}} (1-C_6) (1+C_6)$, by applying the union bound, we obtain that 
    \[
        \pr{\sum_{i\in[n]} Y^{(i)}_1 \geq (1+C_6) \sum_{i\in[n]} Y^{(i)}_j} \geq 1 - \frac{2}{n^4},
    \]
    which implies that $p'_j\leq x p'_1$, with $0<x=\frac{1}{1+C_6}<1$.  Therefore, we obtain that with probability $1-\frac{2}{n^4}$,
    \begin{equation}
        \label{eq:multiplicative_bias_increase}
        \frac{p'_j}{p'_1} \leq \frac{1}{1+C_6} < \frac{p_j}{p_1}
    \end{equation}
    whenever $\sqrt{\frac{2(p_1+p_2)}{h}} \leq \delta(j) < \pa{1 - \frac{1}{1+C_6}}p_1$. Moreover, we obtain from the previous equation that
    \begin{equation}
        \label{eq:bias_increase_2}
        \pr{\delta(j)' \geq \pa{1-\frac{1}{1+C_6}}p_1'} \geq 1 - \frac{2}{n^4}.
    \end{equation}
\end{proof}

Before proving \Cref{lemma:bias_exponential_growth}, we present a claim in which we apply the Bernstein inequality, which will be useful not only in proving the last lemma, but also later on.

\begin{claim}
    \label{claim:bias_concentration_bernstein}
    Let $hp_1 \geq C_4 \log n$.
    If $p_1\geq C_7\frac{\log n}{n}$
    \[
        \pr{ \delta(j)' \leq C_9 \delta(j) \frac{\pr{\WW_1}}{p_1}\sqrt{\log n}  } \leq \frac{1}{n^4},
    \]
    for some constant $C_9>0$.
\end{claim}

\begin{proof}[Proof of \Cref{claim:bias_concentration_bernstein}]
    By \Cref{proposition:bias_hp>1} and by the fact $hp_1\geq C_4 \log n$, we obtain that 
    \begin{equation}    \label{eq:inequality_expected_value_delta}
        \bbE(\delta(j)')\geq C_5 \delta(j) \pr{\WW_1} \sqrt{\frac{h}{p_1}}\geq C_5 \delta(j) \frac{\pr{\WW_1}}{p_1}\sqrt{C_4 \log n}.
    \end{equation}
    Let $X^{(i)}$ be a r.v. that take value 1 if at the next round the agent $i$ adopts opinion 1, value -1 if the agent $i$ adopts opinion $j$ and value 0 otherwise. Notice that $\cpa{X_i}_{i\in[n]}$ are i.i.d., that $ n \delta(j)' = \sum_{i\in [n]} X^{(i)}$, and that
    \begin{align}
        \nonumber
        \Var\pa{X^{(1)}}
        & \leq \bbE\pa{\pa{X^{(1)}}^2}
        \\ \nonumber
        & = \pa{\pr{\WW_1} + \pr{\WW_j}}
        \\        \label{eq:upper_bound_variance_bias}
        & \leq 2 \pr{\WW_1}.
        &\pa{\pr{\WW_j}\leq \pr{\WW_1}}
    \end{align}
    By the Bernstein inequality we obtain
    \begin{align}
        \nonumber
        & \pr{ n \delta(j)' \leq \frac{1}{2} C_5 n \delta(j) \frac{\pr{\WW_1}}{p_1}\sqrt{C_4 \log n}  }
        &\pa{\text{by \Cref{lemma:relationship_ties_without_ties}} }
        \\
        \nonumber
        \leq \ & \, \text{Pr}\left( n \delta(j)' \leq \bbE(n \delta(j)') \right. \\
        \nonumber
        &\left. \quad 
        - \frac{C_5}{2} n \delta(j) \frac{\pr{\WW_1}}{p_1}\sqrt{C_4 \log n} \right) 
        &\pa{\text{by \Cref{eq:inequality_expected_value_delta}}}
        \\
        \nonumber
        \leq \ & \exp\pa{ - \frac{\pa{ \frac{1}{2} C_5 n \delta(j) \frac{\pr{\WW_1}}{p_1}\sqrt{C_4 \log n}}^2}{2 n \Var\cpa{\delta(j)'} + \frac{2}{3}  C_5 \delta(j) \frac{\pr{\WW_1}}{p_1}\sqrt{C_4 \log n}} }        
        &\pa{\text{by \nameref{lemma:bernstein_inequality}}}
        \\
        \nonumber
        \leq \ & \exp\pa{ - \frac{\pa{ \frac{1}{2} C_5 n \delta(j) \frac{\pr{\WW_1}}{p_1}\sqrt{C_4 \log n}}^2}{4 n \pr{\WW_1} + \frac{2}{3}  C_5 n \delta(j) \frac{\pr{\WW_1}}{p_1}\sqrt{C_4 \log n}} }
        &\pa{\text{by \Cref{eq:upper_bound_variance_bias}}}
        \\
        \nonumber
        = \ & \exp\pa{ - \frac{n \pr{\WW_1}\pa{ \frac{1}{2} C_5 \delta(j) \frac{1}{p_1}\sqrt{C_4 \log n}}^2}{4 + \frac{2}{3}  C_5 \delta(j) \frac{1}{p_1}\sqrt{C_4 \log n}} }
        \\
        \nonumber
        \leq \ &  \exp\pa{ - \frac{\frac{n}{3p_1}\pa{ \frac{1}{2} C_5 \delta(j) \sqrt{C_4 \log n}}^2}{4 + \frac{2}{3}  C_5 \delta(j) \frac{1}{p_1} \sqrt{C_4 \log n}} }
        &\pa{\text{by \Cref{lemma:bound_W_1-2}}}
        \\
        \nonumber
        \leq \ &  \exp \left( - \min \left\{ \frac{n}{24 p_1}\pa{ \frac{1}{2} C_5 \delta(j) \sqrt{C_4 \log n}}^2 , \right. \right.
        \\
        & \quad\quad\quad\quad\quad\quad \, \left. \left. \frac{C_5}{4} n \delta(j) \sqrt{C_4 \log n} \right\} \right)
        \\
        \label{eq:amplification_bias_whp}
        \leq \ &  \frac{1}{n^{C}},
    \end{align}
    where the latter holds because \(\delta(j) \geq C_8 \sqrt{\frac{p_1}{n}}\geq \frac{C_8\sqrt{C_7 \log n}}{n}\),
    with $C=\frac{C_4\, {C_5}^2}{24}\geq 4$ for $C_4$ large enough. By dividing both sides by $n$ in the event we are computing the probability, we conclude the proof of \Cref{claim:bias_concentration_bernstein}.
\end{proof}

\begin{proof}[Proof of \Cref{lemma:bias_exponential_growth}]
    By \Cref{lemma:relationship_ties_without_ties} we have $\frac{\pr{\WW_1}}{p_1}\geq \frac{1}{18}$, and therefore, by \Cref{claim:bias_concentration_bernstein}, we obtain
    \[
      \pr{ \delta(j)' \leq \frac{C_9}{18} \delta(j) \sqrt{\log n}  }
        \leq 
        \pr{ \delta(j)' \leq C_9 \delta(j) \frac{\pr{\WW_1}}{p_1}\sqrt{\log n}  } \leq \frac{1}{n^4}.
    \]
    For $n$ large enough, we have $\sqrt{\log n} \geq \frac{18\cdot e}{C_9}$, and therefore, we obtain
    \begin{equation}        \label{eq:multiplicative_factor_bias}
    \pr{\delta(j)' \geq e \, \delta(j)}\geq 1-\frac{1}{n^4}.
    \end{equation}
\end{proof}

\begin{proof}[Proof of \Cref{lemma:new_bias_hypothesis}]
    Firstly, consider the case $\delta(j) \leq \sqrt{\frac{2(p_1+p_2)}{h}}$. By \Cref{claim:bias_concentration_bernstein}, we have with probability $1 - \frac{1}{n^4}$,
    \begin{align*}
        \label{eq:bias_increase}
        \nonumber
        \delta(j)'
        &\geq 
        C_9 \delta(j) \frac{\pr{\WW_1}}{p_1}\sqrt{\log n}
        \\
        \nonumber
        &\geq
        C_9 \,\frac{\pr{\WW_1}\sqrt{\log n}}{\sqrt{p_1 n}} 
        &\pa{\delta(j) \geq \sqrt{\frac{p_1}{n}}}
        \\
        \nonumber
        &\geq 
        C_9 \,\sqrt{\frac{\pr{\WW_1} \log n}{18 n}}
        &\pa{\text{by \Cref{lemma:relationship_ties_without_ties}}}
        \\
        \nonumber
        &\geq
        C_9 \sqrt{\frac{(1-C_6^2/2)p_1' \log n}{n}}&\pa{\text{by \Cref{eq:concentration_p1_next_round}}}
        \\
        &\geq
        C_8 \sqrt{\frac{p_1'}{n}}.&\pa{\text{for $n$ large enough}}
    \end{align*}
    Now consider the case $\delta(j)\geq \sqrt{\frac{2(p_1+p_2)}{h}}$. By \Cref{lemma:disappear_one_step,lemma:disappear_two_steps}, we have that, whenever $\delta(j)\geq \sqrt{\frac{2(p_1+p_2)}{h}}$, 
    \begin{align*}
        \delta(j)'
        &\geq 
        \min\cpa{\pa{1-\frac{1}{1+C_6}} p'_1, p'_1} 
        \\
        &\geq 
        \pa{1-\frac{1}{1+C_6}} \sqrt{p'_1 \frac{C_7 \log n}{n}}
        &\pa{p'_1\geq C_7 \frac{\log n}{n}}
        \\
        &\geq 
        C_8\sqrt{\frac{p'_1}{n}},
        &\pa{\text{for $n$ large enough}}
    \end{align*}
    with probability $1- \frac{2}{n^4}$. These conclude the proof of \Cref{lemma:new_bias_hypothesis}.
\end{proof}

\begin{proof}[Proof of \Cref{claim:all_bias_increase_p1_increase}]
       Let $P_X= \sum_{j\in X} p_j$, and $P'_X= \sum_{j\in X} p'_j$, for $X=I,J,K$. For $j\in I$, by the hypothesis, we have that
       \[
            p'_j < p'_1 - (p_1 - p_j).
       \]
       Hence,
       \begin{equation}
           \label{eq:P'_I}
           P'_I = \sum_{j\in I} p'_j < \sum_{j\in I}p'_1 - (p_1 - p_j) < \abs{I}(p'_1-p_1) + P_I.
       \end{equation}
       For $j\in J$, by the hypothesis, we have that
       \[
            p'_j < p'_1 \cdot \frac{p_j}{p_1}.
       \]
       Hence,
       \begin{equation}
           \label{eq:P'_J}
           P'_J = \sum_{j\in J} p'_j < \sum_{j\in I} p'_1 \cdot \frac{p_j}{p_1} < p'_1 \frac{P_J}{p_1}.
       \end{equation}
       By the fact the probabilities sum to 1, and by \Cref{eq:P'_I,eq:P'_J}, we have
       \[
            1 = p'_1+P'_I+P'_J < p'_1 + \abs{I}(p'_1-p_1) + P_I + p'_1 \frac{P_J}{p_1} = p'_1 (1+\abs{I}+\frac{P_J}{p_1}) - \abs{I}p_1 + P_I.
       \]
       Since $P_J= 1 - p_1 - P_I - P_K > 1 - p_1 - P_I$, this implies that
       \[
            1 + \abs{I}p_1 -P_I < p'_1 \pa{ 1+ \abs{I} + \frac{1 - p_1 - P_I}{p_1} } = p'_1\pa{\abs{I}+\frac{1-P_I}{p_1}},
       \]
       which implies that
       \[
            p'_1 > \frac{1 + \abs{I}p_1 -P_I}{\abs{I}+\frac{1-P_I}{p_1}} = p_1.
       \]
       This concludes the proof of \Cref{claim:all_bias_increase_p1_increase}.
    \end{proof}

\section{Tools}\label{sec:tools}

\begin{lemma}[Multiplicative forms of Chernoff bounds \cite{dubhashi2009}]\label{lemma:app:multiplicative-chernoff}
    Let \(X_1, \ldots, X_n\) be independent binary random variables.
    Let \(X = \sum_{i = 1}^n X_i\) and \(\mu = \expect[X]\).
    Then:
    \begin{enumerate}
        \item For any \(\delta \in (0,1)\) and any \(\mu \le \mu_+ \le n\), it holds that 
        \[
            \Pr(X \ge (1+\delta) \mu_+) \le \exp(- \delta^2 \mu_+ / 3).
        \]
        \item For any \(\delta \in (0,1)\) and any \(0 \le \mu_- \le \mu\), it holds that 
        \[
            \Pr(X \le (1-\delta) \mu_-) \le \exp(- \delta^2 \mu_- / 2).
        \]
    \end{enumerate}
\end{lemma}

\begin{lemma}[Hoeffding bounds \cite{mitzenmacher2005}]\label{lemma:app:hoeffding}
    Let \(a < b)\) be two constants, and \(X_1, \ldots, X_n\) be independent binary random variables such that \(\Pr(a \le X_i \le b) = 1\) for all \(i \in [n]\).
    Let \(X = \sum_{i = 1}^n X_i\) and \(\mu = \expect[X]\).
    Then:
    \begin{enumerate}
        \item For any \(t > 0\) and any \(\mu \le \mu_+\), it holds that 
        \[
            \Pr(X \ge \mu_+ + t) \le \exp(-\frac{2t^2}{n(b-a)^2}).
        \]
        \item For any \(t > 0\) and any \(0 \le \mu_- \le \mu\), it holds that 
        \[
            \Pr(X \le \mu_- - t) \le \exp(-\frac{2t^2}{n(b-a)^2}).
        \]
    \end{enumerate}
\end{lemma}

\begin{lemma}[Bernstein Inequality \cite{dubhashi2009}]
    \label{lemma:bernstein_inequality}
    Let \( X_1, X_2, \dots, X_n \) be i.i.d. random variables such that
    \(\abs{X_i - \mathbb{E}[X_i]} \leq C\) and \(\abs{\expect[X_i]} \le C\) for some \(C > 0\).
    Define the sum of centered variables
    \(
    S_n = \sum_{i=1}^n (X_i - \mathbb{E}[X_i]).
    \)
    Then, for all \( t > 0 \),
    \[
    \Pr(\abs{S_n} \geq t) \leq 2 \exp\left( -\frac{t^2}{2 n \Var(X_1) + \frac{2Ct}{3}} \right).
    \]
\end{lemma}

\end{document}